\documentclass[11pt]{article}

\usepackage[english]{babel}

\usepackage[letterpaper,top=2cm,bottom=2cm,left=3cm,right=3cm,marginparwidth=1.75cm]{geometry}

\usepackage{amsmath}
\usepackage{amssymb}
\usepackage{amsthm}
\usepackage{mathrsfs} 
\usepackage{graphicx}

\usepackage{comment}
\usepackage{setspace}
\usepackage{lipsum}
\usepackage{tikz}
\usepackage{tkz-euclide} 
\usetikzlibrary{patterns}
\tikzset{
    support/.style={
        pattern=north east lines,
        draw=none,
        minimum width=0.3,
        minimum height=0.6
    }
    ,>=latex
}
\usepackage{bbm}
\usepackage{pgfplots}
\pgfplotsset{compat = newest}

\usepgfplotslibrary{fillbetween}
\usetikzlibrary{intersections}

\usepackage{arydshln,leftidx,mathtools}
\usepackage{longtable}
\usepackage{float}
\usepackage{xcolor}
\usepackage{subcaption}

\newtheoremstyle{mystyle}%
  {}%
  {}%
  {\itshape}%
  {}%
  {\bfseries}%
  {.}%
  { }%
  {\thmname{#1}\thmnumber{ #2}\thmnote{ (#3)}}%

\theoremstyle{mystyle}
\newtheorem{theorem}{Theorem}[section]
\newtheorem{lemma}{Lemma}[section]
\newtheorem{prop}{Proposition}
\newtheorem{corollary}{Corollary}
\newtheorem{thm}{Theorem}
\newtheorem{observation}{Observation}[section]

\theoremstyle{definition}
\newtheorem*{remark}{Remark}

\newcommand{\norm}[1]{\lVert #1 \rVert}

\usepackage{subfiles}

\newcommand{\sigto}{\sigma_{\theta_1}}
\newcommand{\sigtt}{\sigma_{\theta_2}}

\usepackage{enumitem}

\usepackage[normalem]{ulem}

\usepackage[title]{appendix}
\newcommand{\leqnomode}{\tagsleft@true\let\veqno\@@leqno}

\usepackage{bm}

\usepackage{tikz}

\usepackage{titling}

\usepackage{pgfplots}
\pgfplotsset{compat=1.17}

\usepackage[authoryear, round]{natbib}

\usepackage{hyperref}
\hypersetup{
    colorlinks=true,
    linkcolor=blue,  %
    urlcolor=red,     %
    citecolor=gray
}

\title{A Costly-information Foundation for Psychometric Curves\thanks{I am deeply indebted to Luciano Pomatto and Axel Niemeyer for their invaluable guidance and support in developing this project. I also thank Ben Wincelberg, Omer Tamuz, Peter Caradonna, Po Hyun Sung, Chew Soo Hong and Song Peicong for their useful comments and encouragement.}
}
\author{Jake Zhang}
\date{Jul 2025}

\usepackage{titlesec}

\usepackage{booktabs}
\begin{document}

\setlength{\abovedisplayskip}{10pt}
\setlength{\belowdisplayskip}{10pt}

\newtheorem{assumption}{Assumption}[section]
\newtheorem{definition}{Definition}[section]

\theoremstyle{definition}
\newtheorem{assum}{Assumption}
\newtheorem{example}{Example}[section]
\AtBeginEnvironment{example}{%
  \pushQED{\qed}\renewcommand{\qedsymbol}{$\triangle$}%
}
\AtEndEnvironment{example}{\popQED\endexample}
\newtheorem{obs}{Observation}
\maketitle
\begin{abstract}

    We study a binary choice problem in which an agent chooses between two actions whose payoff depends on a continuous state. The agent chooses how much effort to invest in learning about the state. Equivalently, we can think of the state as the strength of a stimulus, with the agent exerting costly effort to be more responsive to it. 
    Taking as given the Fisher information cost introduced by \cite{HW2021}, we analyze the optimal state-dependent choice rule using a variational approach.
    The main result is that agents' optimal response is an S-shaped function of the state under mild conditions. This prediction is aligned with the widely documented psychometric curve response profile observed in the experimental literature in psychology and economics.

\end{abstract}


\newpage
\section{Introduction}
Many important decisions require choosing between discrete actions based on imperfectly understood, continuously varying evidence. 
A lender decides whether to approve a loan after assessing the borrower's repayment probability; an investor decides whether to fund a project after assessing its expected return; a physician decides whether to prescribe a particular treatment after assessing the likelihood of presence of a disease; and a judge decides whether to grant probation after assessing the defendant's risk of re-offending. 
In each case, the relevant evidence can be represented by an ordered index, with one action preferred below a critical value and the other preferred above it.  A fully informed decision maker would choose the action with the higher state-dependent payoff, typically according to a cutoff rule.

Actual behavior, however, need not exhibit a sharp cutoff. When nearby values of the index are difficult to distinguish, the probability of choosing an action may vary gradually around the threshold rather than jump discontinuously. Observed choices therefore reflect both what decision makers value and how well they understand the situation they face.
A rejection may indicate that an option is genuinely unattractive, or that its merits were imperfectly assessed. Distinguishing these explanations is important for both welfare analysis and policy: changing incentives may be appropriate in the first case, whereas improving measurement, communication, or information processing may be more effective in the second.

Perceptual experiments in psychology provide a particularly clean way to study this mapping from continuously varying states to discrete choices. Subjects are presented with stimuli that differ along a physical dimension and are asked to give a discrete response: which light is brighter, which tone is higher-pitched, which object is heavier, or which direction a cloud of dots is moving. Data from these experiments are usually summarized by a psychometric curve, whose horizontal axis records stimulus strength and whose vertical axis records the frequency of a particular response. A long-standing empirical regularity is that the psychometric curves are typically sigmoidal.

The same representation can be applied to economic decisions. When the underlying state determines the relative value of two alternatives, the response curve records how the probability of choosing one option changes as that option becomes more attractive.
\citet{C2025} refers to this economic analogue of the psychometric curve as a cognitive economic curve. 
Its shape contains economically relevant information about where mistakes occur, how responsive behavior is to changes in the decision environment, and how successfully decision makers translate available information into choices. Holding incentives and prior belief fixed, changes in the curve can also indicate how the organization or presentation of information affects decision quality.

In this paper, we provide a costly-information explanation for why sigmoidal response curves arise naturally in such environments. We study a binary choice problem in which the decision maker chooses how finely to discriminate among nearby values of a continuous state. The state may represent the strength of a sensory stimulus, but it may equally represent repayment prospects, expected return, likelihood of a disease, re-offending risk, or another ordered determinant of relative payoff. Viewed as a information acquisition problem, the decision maker trades off the benefit of more accurate choices against the cost of information processing.

Building on the neighborhood-based approach of \citet{HW2021}, we use the Fisher information cost to capture the difficulty of discriminating among nearby stimulus levels.
Our main result shows that, under a uniform prior over the state and an increasing relative payoff of one action over the other, any non-degenerate optimal response curve is sigmoidal: the choice probability is increasing, with a slope that first increases and then decreases.

The cost function is based on Fisher information, which measures how sensitively a probability distribution changes with respect to an underlying parameter. In our setting, the parameter is the stimulus level, and the relevant distribution is the subject’s stochastic choice rule at each state. The Fisher information cost aggregates these local sensitivities across the state space, making it costly for choice probabilities to vary sharply across nearby stimulus levels. When the relative payoff of one action is increasing in the stimulus, the optimal response balances this smoothness cost against the benefit of choosing the better action. As a result, the response rule changes most rapidly in the intermediate region where the decision maker switches from favoring one action to favoring the other, while it flattens in regions where the evidence strongly favors one action.

To derive our results, we reparametrize the original variational problem and applies techniques in calculus of variation to examine existence and characterization of the optimal solution. In the binary case, this reparametrization converts the Fisher information term into a quadratic derivative term and yields a tractable characterization of the optimal response rule. In particular, the necessary condition described by the Euler-Lagrange equation can be interpreted as the motion of a particle subject to a time-dependent force, which makes the qualitative shape of the solution transparent.

Beyond the baseline sigmoidality result, the analysis also yields comparative statics with respect to incentives. When payoffs are scaled up, the model predicts that the decision maker chooses a more precise response rule: choice probabilities move closer to zero and one in the tails, and the response curve becomes steeper near the threshold. This prediction connects the model to experimental designs that vary rewards for correct responses and to the psychophysical literature, where sensitivity and lapse rates are commonly used to summarize changes in performance.

\subsection{Related Literature}
The paper contributes to the growing literature on rational inattention that moves beyond Shannon mutual information. 

Following \cite{S2003}, much of the rational inattention literature assumes that information costs are proportional to the reduction in Shannon entropy. This specification is analytically attractive and closely related to multinomial logit choice \citep{MM2015}. It has been assumed in applications to price setting and monetary policy \citep{MW2009,PW2014}, portfolio choice and home-bias \citep{VV2009,M2010}, global games \citep{Y2015}, dynamic choice \citet*{SSM2017} and mental accounting \citep{KM2020}.
For a more complete survey of this literature, see \citet*{MMW2023}.

However, the Shannon entropy cost imposes restrictive behavioral implications. It fails to capture the notion of perceptual distance that similar states are more difficult to distinguish than distant ones. 
\citet*{CDL2022} formalizes this through their Invariance under Compression axiom. The axiom states that the states yielding the same payoffs can be merged into a single state without changing the behavior. \cite{DN2023} reject Invariance under Compression in a perceptual task (their Experiment 4). They find that subjects are better able to distinguish stimuli that are farther apart than those that are closer together, even when making those distinctions leads to the same payoff consequences.

\citet{HW2021} address this limitation by introducing neighborhood-based information costs. Their framework equips the state space with neighborhood structure and assumes cost only occurs when distinguishing states in the same neighborhood. They show that the local cost of discrimination is represented by Fisher information in a continuous-state limit. In binary choice problems, they show that Fisher information costs can generate smooth, monotone response curves with convexity in a lower region and concavity in an upper region.

This paper takes the Fisher information cost derived by \cite{HW2021} as a starting point and studies binary-choice problem in the context of perception task. The contribution is primarily methodological and qualitative. I apply a reparametrization that simplifies the original optimization proble and highlight the correspondence between the necessary condition of the transformed problem and the governing equation of a pendulum system. This physical analogy provides useful intuition for the qualitative behavior of the solution. As an application, I derive sharper shape restriction relative to those obtained in \cite{HW2021}, that the curve changes from convex to concave at a single point. The reparametrization also helps generate comparative statics with respect to incentives.

This work is also related to other approaches that allow information costs to depend on which states are being distinguished. \cite*{PST2023} characterize information costs with constant marginal costs as weighted sums of pairwise Kullback–Leibler divergences, where the weights represent the difficulty of discriminating between particular states. 
More recently, \cite*{BDP2025} introduce a broad class of information costs based on multivariate f-divergences. In a discretized one-dimensional environment, they provide conditions under which the resulting response is S-shaped. Relative to these approaches, we focus on a continuous scalar state space, which allows us to derive global shape restrictions from the associated Euler-Lagrange equation.

Finally, this paper also connects to the psychology literature on psychometric responses in perceptual tasks. 
In binary discrimination or classification tasks, subjects’ response probabilities typically vary smoothly with stimulus strength and the corresponding psychometric function is often fitted using logistic, cumulative Gaussian, Weibull, or other sigmoidal specifications. \citep{WH2001}. Standard explanations derive such curves from noisy internal representations, as in signal detection theory \citep{GS1996}, or from sequential evidence accumulation, as in drift-diffusion models \citep{RM2008}. A related literature on efficient coding studies how limited perceptual resources are allocated across stimuli. While classic work emphasizes information maximization and redundancy reduction \citep{L1981,SO2001}, later population-coding models often use Fisher information to describe representational precision and perceptual discriminability \citep{GS2014,WS2015}.

The present paper is complementary to these approaches. It does not model the internal noise, accumulation process, or coding rule directly. Instead, it derives the state-dependent stochastic choice rule from a information acquisition problem. Thus, the S-shaped response is not imposed as a psychometric specification but arises from optimal information processing. The framework allows us to derive comparative statics with respect to incentives and prior beliefs.

\subsection{Outline of the paper}
The remainder of the paper is structured as follows: Section 2 introduces the binary choice problem within the rational inattention framework, briefly reviews relevant literature, and discusses our modeling choice.
Section 3 presents the formal results about the existence, uniqueness and characterization of the optimal solution through a variational approach.
Section 4 explores the connection between the necessary condition and a classical pendulum system. It also illustrates how this physical analogy can yield behavioral insights through two simple applications. Proofs are relegated to the appendix.

\section{The Model}
A decision maker is choosing between two alternatives $L$ and $R$, whose payoffs depend on the underlying state $\theta \in \Theta$. The state space is assumed to be a compact interval $\Theta=[a,b]$ in $\mathbb{R}$. The decision maker is uncertain about the state and is endowed with full support prior with density $\pi$. Her task is to choose an optimal signal structure that maximizes her expected utility minus the cost of generating the signal structure.\footnote{By a \emph{signal structure} we mean what is often called a \emph{Blackwell experiment}. We use ``signal structure” instead of ``experiment" to avoid confusion with laboratory experiments.} 

Formally, a signal structure consists of a pair $\langle S, \sigma \rangle$. We use $S$ to denote the signal space, which is assumed to be finite.\footnote{So $\Delta(S)$ is a $(|S|-1)$-simplex and the topology on it is the Euclidean topology from $\mathbb{R}^{|S|}$; The topology on $\Theta$ is the Euclidean topology on $\mathbb{R}$.} Meanwhile, $\sigma: \Theta \to \Delta(S)$ is a measurable mapping from the state space to distributions over the signal space, with $\sigma(\theta):=\sigma(\cdot|\theta)$ being interpreted as the signal distribution conditional on the state being $\theta$. For the information cost to be well-defined later, we also assume that the partial derivative $\frac{\partial}{\partial \theta} \sigma(s|\theta)$ exists for each $s$ and for almost all $\theta$. The collection of signal structures that satisfies the above assumption is denoted by $\mathcal{S}$.

The decision maker can flexibly select any signal structure from the collection $\mathcal{S}$ at a cost $C: \mathcal{S} \to [0,+\infty]$. Specifically, we assume that the cost function is the \textit{Fisher information} cost \citep[Section II.C]{HW2021}, which is the averaged Fisher information. %
The Fisher information quantifies informativeness by measuring how much the signal distributions vary across states locally. It is formally defined as the variance of the partial derivative with respect to $\theta$ of the log-likelihood:
\[\mathcal{I}_{F}(\{\sigma(\theta)\}):= 
E_{s|\theta}\left[ \left(\frac{\partial}{\partial \theta}\log{\sigma(s|\theta)} \right)^2 \right]=
\sum_{s \in S} \frac{ (\frac{\partial}{\partial \theta} \sigma(s|\theta))^2}{\sigma(s|\theta)}.\]
The Fisher information cost is:
\[C_{FI}(\sigma):= E_{\theta}[\mathcal{I}_{F}(\{\sigma(\theta)\})]
=\int_{\Theta} \pi(\theta) \mathcal{I}_{F}(\{\sigma(\theta)\} \, d\theta=
\int_{\Theta} \pi(\theta) \sum_{s \in S} \frac{ (\frac{\partial}{\partial \theta} \sigma(s|\theta))^2}{\sigma(s|\theta)} d\theta.
\]

The Fisher information cost satisfies Blackwell monotonicity: a more informative signal structure, in the sense of \citet{B1951}, has a higher cost. We prove this in Appendix~\ref{pf:Blackwell}. By a standard result for Blackwell monotone costs in the spirit of the revelation principle, it is without loss of generality to assume that the optimal signal space is no larger than the action space \citep[Lemma 1]{CD2015}. Therefore, it is without loss of generality to assume the optimal signals are binary, and to interpret a signal as a recommendation to take a specific action.\footnote{That is, we can think of the signal space is $\{R,L\}$, on receiving $R$, taking action $R$ is the best response to her posterior, and vice versa for $L$. So, we can simply read the conditional probability of receiving a signal as the conditional probability of taking a corresponding action on that state.}

We use the short-hand notation $p(\theta):=\sigma(R|\theta)$ to denote the probability of receiving a signal recommending action $R$ in state $\theta$.
We also normalize the state-dependent utility of choosing action $L$ to be $0$ and denote the utility of selecting $R$ in state $\theta$ as $u(\theta)$.\footnote{Such a normalization is without loss of generality when the decision maker has von-Neumann expected utility, because the objective function differs by a constant, that is $\int \pi(t)[ \sigma(R|t) u_R(t) +(1-\sigma(R|t))u_L(t)]=\int \pi(t)[\sigma(R|t) (u_R(t)-u_L(t))+u_L(t)]$.} 
We require $p$ to be absolutely continuous,\footnote{A function $y$ is absolutely continuous on an interval $[a,b]$ if and only if there exists a Lebesgue-integrable function $f$ on $[a,b]$ such that it can be represented as 
 \[ y(t) = y(a) + \int_{0}^{t} f(\tau) d\tau \]
 Hence, an absolutely continuous function $y$ has a derivative almost everywhere.}
 and denote the set of such state-dependent stochastic choice rules by $AC(\Theta,[0,1]):=\{p:\Theta \to [0,1]\, | \,p \text{ is absolutely continuous}\}$.

Throughout this paper, we assume that our payoff function $u$ is piecewise continuous.
Since the state space $\Theta$ is a compact interval, such an assumption also implies that $u$ is bounded and integrable. In addition, another essential assumption is that our payoff function is single-crossing, and without loss of generality, we assume it is single-crossing from below. Both assumptions are quite natural in the lab setting, where subjects are paid for correct answers. Besides, to avoid unnecessary technical complications, we also a mild regularity condition on the prior. Specifically, we assume that the prior admits a density $\pi$ with full support, that $\pi$ is absolutely continuous, and that its reciprocal is integrable.\footnote{The last condition ensures that the change of variables used below to uniformize the prior is well defined.}  
\begin{assum} \label{ass:pwc}
    The payoff function $u$ is piecewise continuous, that is, for $\Theta=[a,b]$, there exists $a = t_0 < t_1 <\cdots t_N = b$, such that 
    \begin{enumerate}
        \item $u$ is continuous on $(t_{i},t_{i+1})$ for $i=0,1,\cdots,N-1$;
        \item left limits $\underset{t \to t_{i}^{-}}{\lim} u(t)$ exist and are finite for $i=1,2,\cdots N$, and right limits $\underset{t \to t_{i}^{+}}{\lim} u(t)$ exist and are finite for $i=0,1,\cdots N-1$.
    \end{enumerate}
\end{assum}
\begin{assum} \label{ass:sc}
    The payoff function $u$ is single-crossing, that is, there exists $t_u$ such that 
    $u(t) < 0 \text{ on $[a,t_u)$ and } u(t) >0 \text{ on } (t_u,b]$.
\end{assum}

\begin{assum} \label{ass:prior}
The prior admits a density $\pi$ on $\Theta=[a,b]$ with respect to Lebesgue
measure. The density satisfies $\pi(t)>0$ for all $t$, $\pi$ is absolutely continuous, and $\int_a^b 1/\pi(t)\,dt<\infty$. 
\end{assum}

In the binary signal setting, the Fisher information cost can be rewritten as:
\[    C_{FI}(p,\pi) =
\int_{\Theta} \pi(t) \frac{(p'(t))^2}{p(t)(1- p(t))}  dt\]
where we treat 
\[ \frac{(p'(t))^2}{p(t)(1- p(t))} =\begin{cases}
         \frac{(p'(t))^2}{p(t)(1- p(t))} , & \text{ if } p(t) \in (0,1) \\
        0, & \text{ if } p'(t) =0 \text{ and $p(t)=$ 0 or 1} \\
        +\infty & \text{ other cases}
    \end{cases}.\]
Specifically, we allow the decision maker to acquire no information at no cost, that is, when $p$ is a constant, thus totally uninformative, the cost is 0.

The expected payoff under the choice rule $p$ is:
\[
\int_{\Theta} \pi(t) u(t) p(t) dt - \lambda C_{FI}(p,\pi),
\]
where $\lambda \in (0,+\infty)$ is a scaling parameter. 

In summary, the decision maker needs to choose a state-dependent stochastic choice rule to solve the following optimization problem:\footnote{We adopt the standard notation from Calculus of Variations and write the problem as a minimization problem} 
\begin{align*}
    \underset{p \in AC(\Theta,[0,1])}{\text{minimize}} \int_{\Theta} \lambda  \pi(t)  \frac{(p'(t))^2}{p(t)(1- p(t))} - \pi(t) u(t) p(t) dt \tag{$P$} \label{pb:p}
\end{align*}

\subsection{Derivation of Fisher information Cost}

To capture the concept of perceptual distance, \citet{HW2021} proposed a new class of cost functions named \textit{Neighborhood-based information cost}, which sums with weight over the local information cost across each neighborhood structure.  
In the special case where the local information is Shannon's entropy, \citet[Appendix D.1]{HW2021} show that the sequence of optimization problems on discrete state space can converge to the optimization with Fisher information cost on a continuum of states\footnote{In the sense that, the value of the optimization problem converges, the solution of the problem converges, and hence the posterior and choice probabilities also converge.}.

Alternatively, \citet[Example 5]{BZ2025} point out that the Fisher information can be derived directly as the limit of a class of cost functions called \textit{Total Information cost}:
\[C_{TI}(\sigma)=\sum_{\theta} \pi(\theta) [\sum_{\theta'} \gamma(\theta,\theta') D_{KL}(\sigma_\theta,\sigma_{\theta'})],\]
by noting that Fisher information is the Hessian of KL divergence. We formalize their heuristic arguments in Appendix \ref{app:tic2fic}. The essential assumptions are a) the state space can be linearly ordered as a compact interval in $\mathbb{R}$; b) cost only occurs between adjacent states under discretization; c) the discrimination coefficients between adjacent states are proportional to the inverse of the square distance, that is
\[\gamma(\theta_{i},\theta_{j}) \propto \frac{\mathbf{1}\{ j =i+1\}}{(\theta_{j} - \theta_i)^2}.\]
From the derivation, we can see that if we allow a varying local discrimination coefficient $\gamma(\theta_{i},\theta_{j}) = \frac{\gamma(\theta_i) \mathbf{1}\{ j =i+1\}}{(\theta_{j} - \theta_i)^2}$, then we can generalize the Fisher information cost into:
\[C_{FI}^{\gamma}(\sigma) = \int_{\Theta} \pi(t) \gamma(t) 
\sum_{s \in S} \frac{(\frac{\partial}{\partial t}\sigma(s|t))^2}{\sigma(s|t)}  dt.\]

\section{Solving The Model}
This section develops the analytical tools needed to solve the optimization problem that arises from the binary choice problem with Fisher information cost. We begin by introducing a reparameterization that significantly simplifies the problem and we establish existence of solution with the aid of such reparameterization.
Using standard techniques from the calculus of variations, we derive the necessary conditions for an optimal solution and reveal a connection between the first-order condition and a classic physical system. Leveraging the problem’s convexity, we establish the uniqueness of any interior optimal solution and thus provide a characterization of the interior optimal solution, conditional on its existence. 

More specifically, we consider the following change of variable: Given $p \in AC(\Theta,[0,1])$, we define $w \in AC(\Theta,[-\frac{\pi}{2},\frac{\pi}{2}])$ as the unique solution to $p(t) = \frac{1+\sin(w(t))}{2}$ for each $t \in \Theta$, in other words, $w(t) = \arcsin(2p(t) - 1)$. Our problem becomes\footnote{To see this, note that, $p'(t)= \frac{1}{2} \cos(w(t)) w'(t)$ and $p(t)(1-p(t)) = \frac{1}{4}(1-\sin^2(w(t))) $}:
\begin{equation}
    \underset{w \in AC(\Theta,[-\frac{\pi}{2},\frac{\pi}{2}])}{\text{minimize}} \int_{\Theta}  \lambda (w'(t))^2  - u(t) \frac{1+\sin(w(t))}{2} \, dt. \tag{$W$} \label{pb:w}
\end{equation}

We say that a solution is \textit{admissible} if it satisfies the constraints of the problem and the objective value is finite. 
The admissible classes of function to Problem (P) and (W) are denoted, respectively, as\footnote{Since $u$ is bounded, the following sets are equivalent:
\scriptsize{\begin{align*} 
    &\{ p \in AC(\Theta,[0,1]): \int_{\Theta} \frac{(p'(\theta))^2}{p(\theta) (1-p(\theta))} d\theta < +\infty \} =\{ p \in AC(\Theta,[0,1]): \int_{\Theta} \frac{(p'(\theta))^2}{p(\theta) (1-p(\theta))} - u(t)p(t) d\theta < +\infty \} \\
   & \{ w \in AC(\Theta,[-\frac{\pi}{2},\frac{\pi}{2}]) : \int_{\Theta} (w'(\theta))^2 d\theta < +\infty\} =\{ w \in AC(\Theta,[-\frac{\pi}{2},\frac{\pi}{2}]) : \int_{\Theta} (w'(\theta))^2 -\frac{u(t) \sin(w(t))}{2}d\theta < +\infty\}
\end{align*}}
We drop the second term for notational simplicity.
} 
\begin{align*}
   \mathcal{D}^{P}& := \left\{ p \in AC(\Theta,[0,1]): \int_{\Theta} \frac{(p'(\theta))^2}{p(\theta) (1-p(\theta))} d\theta < +\infty \right\} \\
\mathcal{D}^{W}& :=\left\{ w \in AC(\Theta,[-\frac{\pi}{2},\frac{\pi}{2}]) : \int_{\Theta} (w'(\theta))^2 d\theta < +\infty \right\}
\end{align*}

We argue that it is equivalent to working on problem (W) because there is a bijection between the two admissible classes that preserves the value of the integral.
The following result justifies our reparametrization by noting that there is a bijection between the admissible classes:
\begin{lemma} \label{lem:bij}
    The transformations \( \Phi(p) := \arcsin(2p - 1) \) and \( \Psi(w) := \frac{1+\sin(w)}{2} \) define a bijection between the two classes of admissible solutions:
\[
p = \frac{1+\sin(w)}{2} \in \mathcal{D}^{P}
\quad \Longleftrightarrow \quad
w = \arcsin(2p - 1) \in \mathcal{D}^{W}.
\]
Moreover, the value of the objective function is preserved under such a reparametrization, that is
\[\int_{\Theta}  u(t) p(t) -  \frac{p'(t)^2}{p(t)(1-p(t))} dt=\int_{\Theta}  u(t) \frac{1+\sin(w(t))}{2} -   (w'(t))^2 dt\] if $p=\frac{1+\sin{(w)}}{2}$ and both are admissible.
\end{lemma} 

\begin{proof}
    See Appendix \ref{pf:lem-bij}.
\end{proof}

\begin{remark}
The transformation $p=\frac{1+\sin(w)}{2}$ is the binary version of the square-root transformation commonly used in information geometry \citep{AN2000}. 
A stochastic choice rule $(p_1(t),\ldots,p_N(t))$ can be viewed as a path on the probability simplex.
Under the square-root map $p_i=\phi_i^2$, an infinitesimal displacement with respect to the parameter is
\[
ds^2 =\sum_{i=1}^N \frac{(p_i'(t))^2}{p_i(t)}dt^2 =
4\sum_{i=1}^N (\phi_i'(t))^2 dt^2.
\]
This embeds the simplex into the unit sphere and converts the Fisher-Rao metric into the standard Euclidean metric on a sphere up to a constant factor. 
In the binary case, the simplex is an interval and the square-root image is a quarter circle. The variable $w$ is an angular coordinate on this circle.
\end{remark}

\noindent \textbf{Notation} For notational simplicity, we set the state space to be $\Theta=[0,1]$, the scaling parameter $\lambda$ to be $1$, and assume a uniform prior.  It is without loss of generality to adopt such normalization as we can always recover the solution to the original problem through the solution to the normalized one, the detailed argument are provided in Appendix~\ref{pf:normalize}.

We also introduce the following notation. Let $J: AC([0,1]) \to \mathbb{R}$ be the integral functional $x \mapsto J[x]: = \int_{0}^{1} \Lambda\bigl(t, x(t), x'(t)\bigr)\,dt$. The integrand $\Lambda: \mathbb{R}^3 \to \mathbb{R}$ is called the \emph{Lagrangian}. We use subscripts to denote partial derivatives and superscripts to distinguish between optimization problems \eqref{pb:p} and \eqref{pb:w}. That is:
\begin{align*}
    & \Lambda_{x}(t,x,v) := \frac{\partial}{\partial x}\Lambda(t,x,v) && \Lambda_{v}(t,x,v) := \frac{\partial}{\partial v}\Lambda(t,x,v)\\
    & J^P[p] := \int_{\Theta} \Lambda(t,p(t),p'(t)) dt && \Lambda^{P}(t,p(t),p'(t)) :=    \frac{p'(t)^2}{p(t)(1-p(t))} - u(t) p(t)    \\
    & J^W[w] := \int_{\Theta} \Lambda(t,w,w') dt   && \Lambda^{W}(t,w(t),w'(t)) :=   (w'(t))^2 -  u(t) \frac{1+\sin(w(t))}{2} 
\end{align*}
A solution refers to an admissible function $x_*$ that satisfies $J[x_*] \leq J[x]$ fo all admissible functions $x$, it is also called as a minimizer.

\subsection{Existence of minimizer}
Our first question of interest is whether an optimal solution to Problem \eqref{pb:p} exists in our chosen function class $AC(\Theta,[0,1])$. At first sight, it seems challenging to establish the existence of a solution to Problem \eqref{pb:p} because its Lagrangian $\Lambda^P$ is not coercive. Luckily, the transformation above helps mitigate this difficulty. The coercivity of $\Lambda^{W}$ allows us to apply the standard technique of \textbf{direct method} to establish the existence of a solution in the class of absolutely continuous functions. We then combine our necessary condition to argue that the optimal solution is actually Lipschitz.\footnote{A function $f:[a,b] \to \mathbb{R}$ is said to be Lipschitz continuous if there exists a constant $K>0$ such that
\[
|f(t_1) - f(t_2)| \leq K|t_1 -t_2| \quad \text{ for any } t_1,t_2 \in [a,b].
\]}

We might also care whether our constraint $0 \leq p \leq 1$ is binding. It turns out that, for the optimal solution, the constraint is either always binding or not binding at all, but it can not be active on some intervals and inactive on other intervals. That is, under the optimal signal structure, the agent either learns nothing about the state and sticks to one action or she never chooses an action with certainty under whichever state. This helps us to simplify the necessary condition, so we do not need to impose an extra multiplier for the constraint.

The following statement summarizes the result, and the proof can be found in the appendix: 
\begin{thm} \label{thm:ext}
    Suppose $u$ is piecewise continuous, then Problem \eqref{pb:p} admits a solution in $Lip([0,1],[0,1])$. Furthermore, if $u$ is also single-crossing, then the optimal solution to Problem \eqref{pb:p} belongs to one of the following two kinds:
    \begin{itemize}
        \item Pure boundary solution: $p(t)=0$ for all $t \in [0,1]$ or $p(t)=1$ for all $t \in [0,1]$
        \item Strict interior solution: $0<p(t)<1$ for all $t \in [0,1]$
    \end{itemize}
\end{thm}
\begin{proof}
    See Appendix~\ref{pf:thm-ext}
\end{proof}

\subsection{Uniqueness of interior minimizer}
We now study the uniqueness property of the solution to our problem. 
We have the following observation that the integrand of the minimization Problem \eqref{pb:p} is convex in $(p,p')$.\footnote{The integrand of the transformed problem \eqref{pb:w} is not convex in $(w,w')$, and this is not contradictory. Consider two candidates $w_1,w_2$ and their corresponding transformation $p_1=\frac{1+\sin{(w_1)}}{2},p_2=\frac{1+\sin{(w_2)}}{2}$, the transformation of the convex combination is in general not the convex combination of the transformations, i.e. $\frac{1+\sin{(\lambda w_1 + (1-\lambda)w_2)}}{2} \neq \lambda\frac{1+\sin{(w_1)}}{2} + (1-\lambda)\frac{1+\sin{(w_2)}}{2}$}
\begin{lemma} \label{lem:uni}
The integrand $\Lambda^{P}(t,p,p')= \frac{p'(t)^2}{p(t)(1-p(t))} - u(t)p(t)$ is  convex in $(p,p')$ for each $t$.
\end{lemma}

\begin{proof}
    See Appendix \hyperref[pf:lem-uni]{\ref{pf:lem-uni}}.
\end{proof}

Specifically, the inequality is strict in the convex relation whenever at least one of the derivatives is non-zero, so the functional $J^P$ satisfies: $$J^P[\lambda p_1 + (1-\lambda) p_2] < \lambda J^P[p_1] + (1-\lambda) J^P[p_2] \text{ for } \forall \lambda \in(0,1)$$ whenever $p'_1$ or $p'_2$ is nonzero on a set of positive measure.
Suppose two different solutions achieve the minimum and at least one of them is non-constant. The two solutions must differ on a neighborhood of positive measure, which implies that their convex combination would achieve a lower integral value, a contradiction. 
\begin{corollary}[Uniqueness]
    If the Problem \eqref{pb:p} admits a non-constant interior optimal solution, then it is the unique solution.
\end{corollary}

\subsection{Necessary conditions}
As discussed above, the global minimizer to Problem~(\ref{pb:w}) can either be a pure boundary solution or a strict interior solution. To describe the interior solution, it suffices to examine the necessary condition for the extremals in the problem without state constraints. It turns out that this condition is also satisfied by the pure boundary extremal. Hence, we have the following result:

\begin{thm}[Necessity]  \label{thm:nc2}
    Suppose $u$ is piecewise continuous on $\bigcup_{i=0}^{N}(t_i,t_{i+1})$.
    Let $w_{*}$ be a solution to Problem \eqref{pb:w}, then $w_{*}$ satisfies the \textbf{integral Euler equation} together with \textbf{natural boundary conditions} and is continuously differentiable at each discontinuity of $u$: \phantomsection
    \begin{equation}
       \begin{cases}
           w'_*(t) = \int_{0}^{t} -\frac{1}{4} u(s) \cos(w_*(s)) ds & \text{a.e. on } [0,1]\\
            w_*'(0) = w_*'(1) =0 
       \end{cases} 
       \tag{NC}
       \label{eq:Euler-NC2}
    \end{equation}     
\end{thm}
\begin{proof}
    See Appendix \ref{pf:thm-nc2}.
\end{proof}
\begin{remark}
    We might also write the integral equation as a second-order ordinary differential equation, in the sense of \textbf{Carath\'{e}odory}\footnote{A function $y$ is a solution to a differential equation in the sense of \textbf{Carath\'{e}odory} if $y$ is absolutely continuous and satisfies the differential equation almost everywhere.} \citep[\S 10]{W1998}:
    \[ w''(t) =  -\frac{1}{4} u(t) \cos(w_*(t)) \quad \text{ for almost all t.}\]
\end{remark}
\begin{remark}
    Our necessary condition is not sufficient as it is always satisfied by $w \equiv -\frac{\pi}{2}$ and $w \equiv \frac{\pi}{2}$. This is not too surprising, because for the necessary condition to be sufficient, the literature typically requires the Lagrangian $\Lambda(t,x,v)$ to be convex in both $(x,v)$, but $\Lambda^{W}(t,w,w')$ is only convex in $w'$.
\end{remark}
Our necessary condition is related to a version of the \textbf{Problem of Bolza}, and we include its proof in the Appendix for completeness. We will first discuss its implications for the regularity of the solution and then its interpretation in the next sub-section. 

\subsubsection{Regularity from Euler equation}
Standard results in the calculus of variations literature typically study conditions under which an absolutely continuous solution is also Lipschitz continuous.
The above integral Euler equation offers an easy way to prove higher order regularity for the solution and helps us to understand why Fisher information cost induces smoothness of the optimal solution. In particular, under the mild assumption that payoff is piecewise continuous, the optimal solution is twice differentiable almost everywhere. As a result, the response curve exhibits no kinks or spikes.  

\begin{thm}[Regularity] \label{thm:reg}
Let $p_{*}$ be a solution to Problem~(\ref{pb:p}). If $u$ is piecewise continuous, then $p_{*}$ is twice differentiable almost everywhere.
Furthermore, if $u \in C^{k}$ for some $k \in \mathbb{N}$, then $p_{*} \in C^{k+2}$.  
\end{thm}
\begin{proof}
    See Appendix \ref{pf:thm-reg}.
\end{proof}

\noindent \textbf{Numerical Tractability}: 
The boundary‐value problem arising from the necessary conditions can be solved via a shooting method: one integrates the ODE with  $w'(0) =0$ while varying the initial value $w(0)$. In doing so, the task reduces to a one‐dimensional root‐finding problem. In future work, we will examine the well‐posedness of this initial‐value formulation to rigorously establish its numerical traceability.

\subsection{Connection to a pendulum system:} The necessary condition induces a 2nd order pendulum-type ordinary differential equation with boundary value constraints
\[
\begin{cases}
    w''(t) = -\frac{1}{4} u(t) \cos(w(t)) \text{ a.e.}\\
    w'(0) = w'(1) =0
\end{cases}
\]

This system admits a physical interpretation as a unit-mass object attached to a pivot by a rigid, unit-length rod, and subject to a time-varying external force acting in the horizontal direction.\footnote{We assume there is no gravitational force acting on the object.} The interval $[0,1]$ represents time, with $w(t)$ denoting the angle between the rod and the vertical axis at time $t$, and $\tilde{u}(t)=\frac{1}{4} u(t)$ representing the applied force. Angles to the left of the vertical midline are taken to be negative, and horizontal forces directed to the right are taken to be negative.

A standard force decomposition and Newton’s second law imply that the object experiences a tangential acceleration of $\tilde{u}(t)\cos(w(t))$ along its circular trajectory. The boundary conditions $w'(0)=w'(1)=0$ then translate into the requirement that the object starts and ends at rest, i.e., with zero angular velocity at both endpoints. Figure~\ref{fig:pendulum} illustrates the setup. 

\begin{figure}[H] 
    \centering
 \begin{tikzpicture}[scale=2.5] 
        \coordinate (A) at (-1.5,0);
        \coordinate (B) at (1.5,0);
        \coordinate (C) at (0,2);
        \draw[->,dashed,gray] (-2,0) -- (2,0);
        \draw[->,dashed, gray] (0,2) -- (0,-2) node[coordinate] (O) {};
        \draw (-0.5,2) -- (0.5,2);
        \draw[support] (-0.5,2)  rectangle++ (1,0.25);

        \draw[line width=2pt, gray!60] (C) -- (A);
        \draw[line width=0.5pt, black] (C) -- (A);

        \draw[->] (A) -- ++(0:1.25) node[above right]{{$\hspace{-7pt}\tilde{u}(t)$}};

        \draw[->] (A) -- ++(-36:1) node[below ]{{$\tilde{u}(t)\cos(\theta)$}};
        \draw[->] (A) -- ++(54:0.75) node[above left]{\tiny{$\empty$}};

        \draw[->] (A) -- ++(-126:0.75) node[right,rotate=-30]{\tiny{$\empty$}};

        \draw[dashed,gray] ++(A) ++(-36:1) -- ++(54:0.5);
        \draw[dashed,gray] ++(A) ++(54:0.75) -- ++(-36:1);

        \draw[dashed] (-1.5,0) .. controls (-1,-0.5) and (1,-0.5) .. (1.5,0);
        \draw[fill = black] (A) circle (4pt);
        \draw[dashed] (B) circle (4pt); 

        \tkzMarkAngle[-, size=0.5](A,C,O)
        \tkzLabelAngle[pos=0.6](A,C,O){{$\theta$}}
    \end{tikzpicture}
    
    \caption{A pendulum with varying force}
    \label{fig:pendulum}
\end{figure}
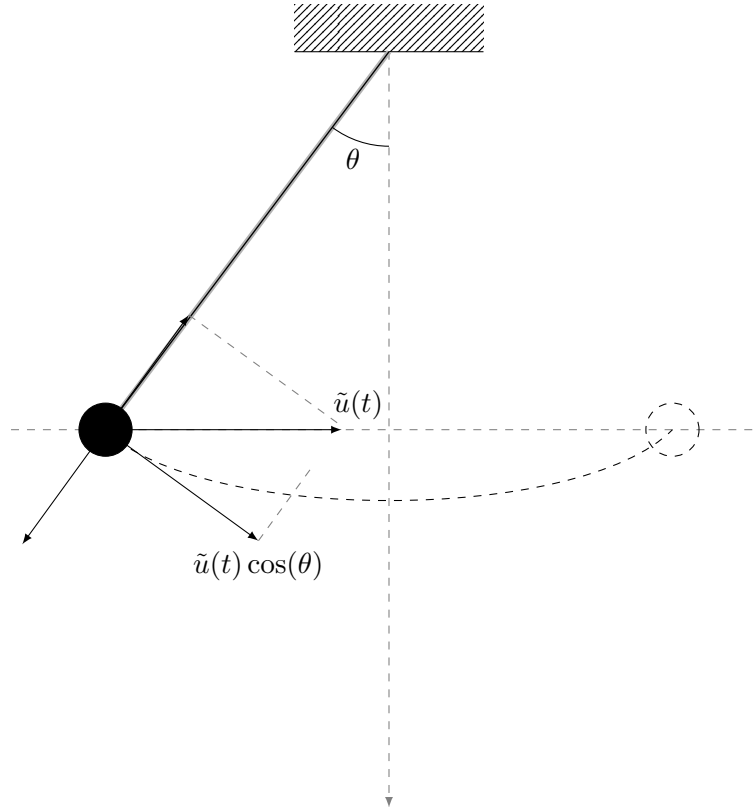

\section{Application}
This physical interpretation provides useful intuition for anticipating the qualitative shape of the solution and analyzing comparative statics. For instance, consider the case where the payoff is strictly positive or strictly negative, corresponding to a force that never changes sign. In such scenarios, bringing the object to rest at the final time requires it to start at one of the boundary angles, either $w = -\frac{\pi}{2}$ or $w = \frac{\pi}{2}$. This results in a constant solution—$w(t) \equiv -\frac{\pi}{2}$ or $w(t) \equiv \frac{\pi}{2}$, which implies that no information is acquired.

Now consider a more interesting case in which the payoff $u(t)$ is \emph{single–crossing from below} (Assumption~\ref{ass:sc}). This corresponds to a force that first pushes to the left and later reverses direction to the right. In this case, the object rotates counterclockwise, first speeds up and then slows down, but never reverses its direction of rotation. Because if it reverses to rotate clockwise when the force is pulling to the left and acting in the same direction, then it would not end up at rest. As a consequence, the angle $w(t)$ must evolve in a single direction over time, hence is monotone. 

\begin{prop} \label{prop:monotone}
    If $u$ is single–crossing from below, then the solution to Problem~\eqref{pb:p} is increasing. Likewise, if $u$ is single-crossing from above, then the solution is decreasing.
\end{prop}

\begin{proof}
  See Appendix~\ref{pf:monotone}.
\end{proof}

\begin{remark}\label{rem:mono_vs_MI}
Proposition~\ref{prop:monotone} only requires the payoff $u$ to be single-crossing, and allows weak inequalities\footnote{That is, there exists $t_{0}\in[0,1]$ such that
$u(t)\le 0$ for $t\le t_{u}$ and $u(t)\ge 0$ for $t\ge t_{u}$.}. This is in contrast to the mutual information cost function, which requires the payoff to be monotone as well.
This is because the Fisher information cost cumulatively captures the cost of distinguishing nearby states and penalizes local variation of the response curve. 

By contrast, under the mutual-information cost, each state is treated separately and the cost is the cumulation of reduction in entropy at each state. As a result, having a higher precision\footnote{By higher precision, we mean the probability of choosing the more likely action is higher. In the binary case, the entropy $H(p)=-p\log(p)-(1-p)\log(1-p)$ is decreasing with $p$ on $[1/2,1]$.} 
at a state always incurs a higher cost with mutual information cost, while this may not be the case for Fisher information cost if the signal structure also has high precision at the states nearby. 
\end{remark}

\subsection{Predicting the psychometric curve}
As an application in predicting the shape of the solution, we provide two sets of sufficient conditions on the payoffs to induce a sigmoidal profile for the choice probability $p$ in Problem~\eqref{pb:p}, which is widely observed in perceptual decision-making experiments.

First, consider the case where the payoff is single-crossing and symmetric in time. In particular, suppose the external force $u$ is antisymmetric around $t = 0.5$, i.e., $u(t) = -u(1 - t)$ and $u(t) <0$ on $[0,\frac{1}{2})$. In such a setting, we can show that Problem~\eqref{pb:w} admits a local interior minimizer. Specifically, by choosing an initial angle sufficiently close to the horizontal (i.e., near the ceiling), the object would cross the vertical midline at $t = 0.5$, so that the angle changes sign together with the force. The time and angular symmetry then imply that the total impulse from the external force cancels out, allowing the object to come to rest at the endpoint.
 This would induce a sigmoidal profile for the choice probability $p(t)$ in Problem~\eqref{pb:p}—a pattern widely observed in perceptual decision-making experiments.

\begin{prop} \label{prop:S-shape1}
    Suppose the prior is uniform over $\Theta$, the payoff is antisymmetric and single-crossing, i.e. $u(t) = -u(1-t) \text{ on } t \in [0,1]$ and $u(t)<0$ on $[0,1/2)$, then the optimal solution is interior. 
    Let $p_{*}$ be the interior maximizer of Problem \eqref{pb:p}. Then: 
    \[p_{*}''(t) >0  \text{ on } [0,1/2), \text{ and }  p_{*}''(t) <0 \text{ on } (1/2,1].\] 
\end{prop}

\begin{proof} 
See Appendix \ref{pf:s-shape1}
\end{proof}

Our second set of conditions only requires the payoff function to be monotone.
The proof identifies subintervals where the second derivative $p''$ has fixed sign (positive on the left, negative on the right), and then shows $p''$ decreases on the intermediate interval by examining the third derivative (in the sense of distributions if not everywhere differentiable). Hence $p''$ changes sign exactly once, and $p'$ is unimodal.

\begin{prop}\label{prop:S-shape2}
    Suppose the prior is uniform over $\Theta$, the payoff $u$ is increasing, piecewise-continuous, and an interior optimal solution $p^*$ exists for the original decision problem (p), 
    then, the optimal solution $p^*$ exhibits sigmoidal shape, that is, there exists a $\tau \in (0,1)$ such that $p'(t)$ is increasing on $[0,\tau)$ and decreasing on $(\tau,1]$.
\end{prop}

\begin{proof}
    See Appendix \ref{pf:s-shape2}
\end{proof}

\begin{example} \label{eg:2AFC}
In a standard two–alternative forced–choice perception task, a subject is presented with two stimuli (e.g., lights, tones, or weights) and asked which is stronger along the relevant dimension (brighter, higher, heavier). The difference in stimulus strength can be considered as the state. If differences are sampled uniformly, it might be safe to model the prior over the state as uniform. Suppose the subject receives a fixed payoff \(U>0\) for a correct response. Then the payoff can be written as
\[
u(t)=
\begin{cases}
-\,U, & t\in[0,\frac{1}{2}),\\[2pt]
\phantom{-}U, & t\in(\frac{1}{2},1],
\end{cases}
\]
This \(u\) is antisymmetric and single-crossing so the assumptions of Proposition~\ref{prop:S-shape1} hold; consequently, the optimal choice probability \(p^*(t)\) is sigmoidal.
\end{example}

\subsection{Predicting higher incentive increase precision}
As another application of the physical intuition, we establish a simple comparative static: scaling up payoffs (equivalently, scaling down information costs) increases the decision maker’s probability of choosing the correct action.

Although the conclusion is intuitive, the pendulum analogy also provides a helpful way to visualize why it must be true and how a proof can proceed. Under our symmetry assumption, the boundary-value problem can be viewed as selecting an initial position for the object so that it passes through the vertical midline at time $t=\frac{1}{2}$, i.e., $w(\frac{1}{2})=0$. When the force is scaled up, the object must start farther from the vertical position. 
Equivalently, it can be considered as starting from time $t=0.5$ at $w(\frac{1}{2})=0$ and choose an intermediate speed $w'(\frac{1}{2})$ so that the object comes to rest at $t=1$. Under the stronger pullback, to travel a larger angular distance naturally requires a larger intermediate speed.

\begin{prop} \label{prop:scale-u}
    Suppose the prior is uniform over $\Theta$, the payoff is antisymmetric and single-crossing, i.e. $u(t) = -u(1-t) \text{ on } t \in [0,1]$ and $u(t)<0$ on $[0,1/2)$.
    Let $0<\lambda_1<\lambda_2$ be two (cost) scaling parameters of Problem \eqref{pb:p}, and 
    $p_{*,1}, p_{*,2}$ be the corresponding maximizer. Then: 
    \[p_{*,1}(0) < p_{*,2}(0) <\frac{1}{2} <p_{*,2}(1) <p_{*,1}(1), \text{ and }  0<p'_{*,2}(\frac{1}{2}) < p'_{*,1}(\frac{1}{2}),\] 
    moreover, 
    \[p_{*,1}(t) < p_{*,2}(t) <\frac{1}{2} \text{ on } [0,\frac{1}{2}), \text{ and } \frac{1}{2} < p_{*,2}(t) <p_{*,1}(t)\text{ on }(\frac{1}{2},1]. \]
\end{prop}

\begin{proof}
    See Appendix \ref{pf:prop-scale-u}
\end{proof}

\subsubsection{Mapping predictions into psychometric parameters}
The implication also connects with how psychophysics typically summarizes perception data.
Choice frequency data from perception experiments are often fit parametrically by a strictly increasing sigmoid function with lapse rate:
\[p(t) = \gamma + (1- \gamma - \kappa) F(t;\alpha,\beta),\] 
where the \textit{upper lapse rate} $(\kappa)$ captures the probability of an incorrect response even at high stimulus levels, and the \textit{lower lapse rate} $(\gamma)$ captures stimulus-independent errors in the opposite tail (and is often called the guess rate when the stimulus is nonnegative, i.e. reporting `Yes' in the absence of a stimulus). The \text{shift parameter} $(\alpha)$ captures the threshold stimulus level where the choice probability $(p)$ reaching some specific level, e.g. $50\%$ and the \text{scale parameter} $(\beta)$ measures the sensitivity of the choice to the stimulus.
A standard specification is
\[p(t) = \gamma + (1- \gamma - \kappa) F\left(\frac{t-\alpha}{\beta}\right),\] 
where $(F)$ is the CDF of a Gaussian or logistic distribution. The result shows that scaling up incentives decreases the lapse rate and increases sensitivity.

To illustrate how this prediction would translate into the parameters commonly reported in psychophysics, we conduct a simulation exercise in which the predicted optimal response curve is treated as the data-generating process. 
We consider the two-alternative forced-choice environment in Example \ref{eg:2AFC}, with a uniform prior and a symmetric step payoff
\[
u(t)=
\begin{cases}
-\,U, & t\in[0,\frac{1}{2}),\\[2pt]
\phantom{-}U, & t\in(\frac{1}{2},1],
\end{cases}
\]
where $(U>0)$ indexes the incentive for a correct response. For each value of $U \in \{20,40,80,160\}$, we solve for the optimal response rule $(p^{*}_{U}(t))$.
We then generate binomial responses-count data using ($p^*_U(t)$) as the success probability with $N=100$ trials at a grid of 19 stimulus levels. On the grid of stimulus levels, we fit by maximum likelihood the logistic specification 
\[
\widehat p(t)=\gamma+(1-\gamma-\kappa)
\frac{1}{1+\exp\left(-(t-\alpha)/\beta\right)}.
\]
Here ($\alpha$) is the threshold parameter, ($\beta$) is the scale parameter, so that lower ($\beta$) corresponds to higher sensitivity, and ($\gamma$) and ($\kappa$) are the lower and upper lapse rates. The fitted parameters and Root Mean Squared Error (RMSE) with $90\%$ bootstrap confidence interval are reported in Table \ref{tab:logit-fit}.

\begin{table}[H]
    \centering
    \caption{Bootstrap Estimates of Psychometric Parameters under Scaled Incentives}
    \label{tab:logit-fit}
    \begin{tabular}{lcccc}
        \toprule
        & \multicolumn{4}{c}{Incentive Level ($U$)} \\
        \cmidrule(lr){2-5}
        Parameter & $U=20$ & $U=40$ & $U=80$ & $U=160$ \\
        \midrule
        
        Sensitivity ($1/\beta$) 
        & 7.807 & 9.050 & 11.598 & 15.842 \\
        & [4.323, 11.779] & [6.957, 11.542] & [9.802, 13.688] & [13.683, 18.305] \\
        \addlinespace
        
        Threshold ($\alpha$)    
        & 0.498 & 0.500 & 0.500 & 0.500 \\
        & [0.436, 0.558] & [0.473, 0.528] & [0.484, 0.517] & [0.489, 0.511] \\
        \addlinespace
        
        Lower Lapse ($\gamma$)  
        & 0.208 & 0.097 & 0.031 & 0.009 \\
        & [0.105, 0.275] & [0.048, 0.139] & [0.005, 0.054] & [0.000, 0.021] \\
        \addlinespace
        
        Upper Lapse ($\kappa$)  
        & 0.209 & 0.096 & 0.031 & 0.009 \\
        & [0.107, 0.275] & [0.044, 0.140] & [0.008, 0.055] & [0.000, 0.022] \\
        \midrule
        
        RMSE                    
        & 0.040 & 0.035 & 0.029 & 0.022 \\
        & [0.029, 0.054] & [0.025, 0.047] & [0.020, 0.040] & [0.014, 0.032] \\
        \bottomrule
        \multicolumn{5}{p{14cm}}{\footnotesize \textit{Notes:} Point estimates represent the mean across 1,000 bootstrap iterations. 90\% confidence intervals [5th, 95th percentiles] are reported in brackets below each estimate. Simulated data consists of 100 binary choices generated at 19 uniformly spaced interior stimulus levels ($t \in [0.05, 0.95]$) for each condition. The true response probabilities were derived numerically from the Fisher information boundary-value problem.}
    \end{tabular}
\end{table}

 The finite-sample exercise shows how the model’s structural comparative statics would be recovered from the type of choice-frequency data observed in a laboratory psychophysics experiment. As incentives increase, the fitted sensitivity measure ($1/\beta$) increases and both lapse-rate estimates decrease while the threshold remains approximately fixed in the symmetric environment.
 Figure~\ref{fig:scale-u} illustrates these comparative statics and shows that the fitted logistic curves closely approximate the structural Fisher curves.

\begin{figure}[H]
    \centering
    \includegraphics[width=\linewidth]{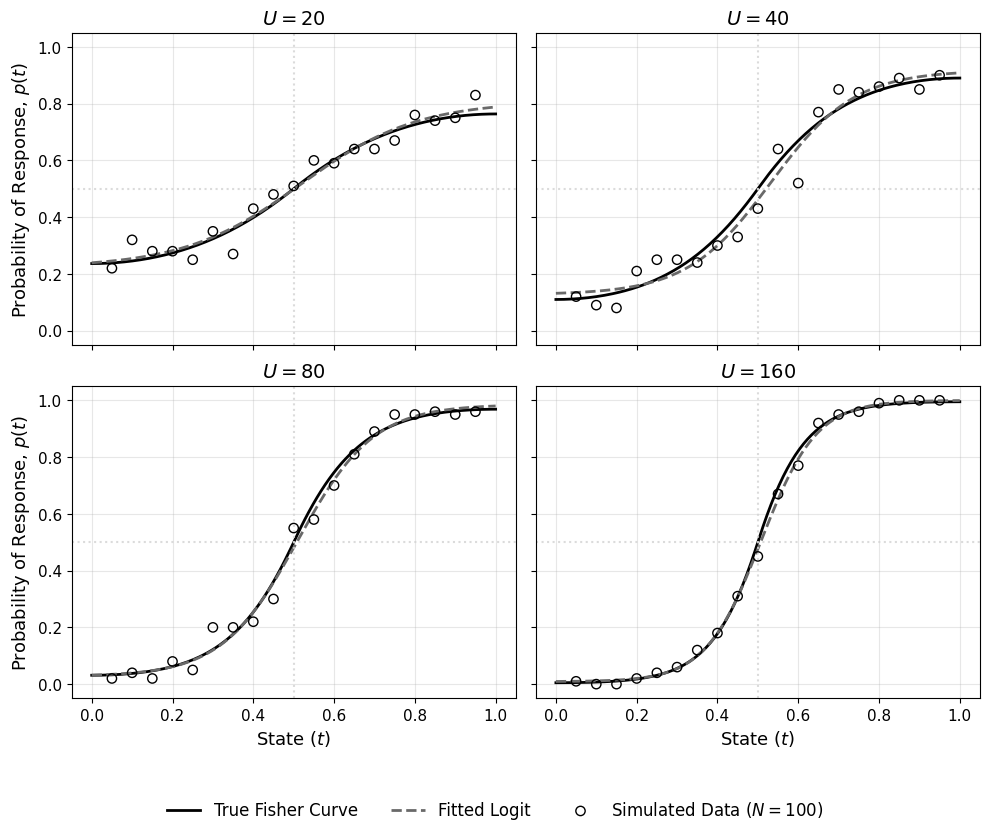}
    \caption{Optimal Response Curves and Simulated Experimental Fits}
    \label{fig:scale-u}
\end{figure}

\bibliographystyle{plainnat}  %
\bibliography{References} 

\newpage
\appendix

\section{Proof of technical lemmas and theorems}
\subsection{Proof of Lemma \ref{lem:bij} (Bijection)} \label{pf:lem-bij}
\begin{proof}
If $\bm{ w\in AC(\Theta,[-\frac{\pi}{2},\frac{\pi}{2}])}$ and $\bm{\int_{\Theta} (w'(t))^2 dt<\infty}$:

Because $\sin(\cdot)$ is a Lipschitz function and the composition of a Lipschitz function with an absolutely continuous function is again absolutely continuous. We know $p(t)=\frac{1+\sin(w(t))}{2}$ is absolutely continuous on $\Theta$. Meanwhile, $p(t)\cdot  (1-p(t))=\frac{1-\sin(w(t))^2}{4}=\frac{\cos(w(t))^2}{4}$ and $p'(t)=\cos(w(t)) \cdot w'(t)$ is defined almost everywhere on $\Theta$. The following equality always hold: $\frac{p'(t)^2}{p(t)(1-p(t))}=w'(t)^2$.
Hence, $\int_{\Theta} \frac{p'(t)^2}{p(t)(1-p(t))} dt =\int_{\Theta} (w'(t))^2 dt<\infty$ and $p \in \mathcal{D}^{P}$.

If $\bm{p \in AC(\Theta,[0,1])}$ and $\bm{\int_{\Theta} \frac{p'(t)^2}{p(t)(1-p(t))} dt<\infty}$: 

Case I: if $p$ is uniformly strictly bounded away from $0$ and $1$, that is, there exists some $0<\epsilon<\frac{1}{2}$ such that $p(t) \in [-1+2 \epsilon, 1-2 \epsilon]$ for all $t \in [0,1]$. Because $\arcsin(\cdot)$ is Lipschitz on $[-1+2 \epsilon, 1-2 \epsilon]$ and the composition of a Lipschitz function with an absolutely continuous function is also absolutely continuous. We know $w$ is absolutely continuous. Meanwhile, $w'(t)=\frac{2p'(t)}{\sqrt{1-(2p(t)-1)^2}}=\frac{p'(t)}{\sqrt{p(t)(1-p(t))}}$ is defined almost everywhere on $\Theta$. 
Hence, $\int_{\Theta} (w'(t))^2 dt=\int_{\Theta} \frac{p'(t)^2}{p(t)(1-p(t))} dt <\infty$ and $w \in \mathcal{D}^{W}$.

Case II: if $p$ takes the value $0$ or $1$. The derivative of $\arcsin(\cdot)$ is not defined at $0$ or $1$, we will truncate and use the property of Sobolev space. The \textbf{Sobolev space} $W^{1,2}(\Theta)$ is defined to be the space of functions that have finite $L^2$ norm and whose \textbf{weak derivative} has finite $L^2$ norm as well:
\[
W^{1,2}(\Theta) = \{f: f \in L^p(\Theta), \text{ and } \exists g \in L^p(\Theta) s.t. \int_{\Theta} f \phi = -\int_{\Theta} g \phi', \forall  \phi \in C_c^1(\Theta) \}
\]
and it is equipped with the norm $\norm{f}_{W^{1,2}} = (\norm{f}_{L^2}^2 +\norm{f'}_{L^2}^2)^{\frac{1}{2}}$.
Equipped this metric, it is a complete space \citep[Proposition 8.1]{B2011}. Moreover, on a bounded interval $\Theta$, every $f \in W^{1,2}(\Theta)$ has a absolutely continuous representation \citep[Theorem 8.2]{B2011}.\footnote{That is, there exists a absolutely continuous function $\tilde{f}$ such that $f=\tilde{f}$ almost everywhere.} 

We truncate $p$ to be $p_{\epsilon} = \max\{ \min\{p,1-\epsilon\} ,\epsilon\}$ and $w_\epsilon = \arcsin(2p_{\epsilon}-1)$. According to Case I, $w_{\epsilon}$ is absolutely continuous. Moreover, $w_{\epsilon}$ is a constant when $p<\epsilon$ or $p>1-\epsilon$, so $w'_{\epsilon}(t) = \frac{p'(t)}{\sqrt{p(t)(1-p(t))}} \mathbf{1}_{\{\epsilon<p(t)<1-\epsilon \}}$ a.e.. Hence, $\int_{\Theta} (w'_{\epsilon}(t))^2 dt=\int_{\Theta} \frac{p'(t)^2}{p(t)(1-p(t))} \mathbf{1}_{\{\epsilon<p(t)<1-\epsilon \}}dt \leq \int_{\Theta} \frac{p'(t)^2}{p(t)(1-p(t))} dt <\infty$.

Now we consider a sequence $\epsilon_n \to 0$, we will show that $w_{\epsilon}$ is a Cauchy sequence in $W^{1,2}(\Theta)$ and the limit of the sequence is $w=\arcsin(2p-1)$. Firstly, we notice that $p_{\epsilon}$ converges to the un-truncated $p$ point-wise, and together with continuity of $\arcsin(\cdot)$, $w_{\epsilon}$ also converges point-wise to $w$. So $\norm{w - w_{\epsilon_n}}_{L^2} \to 0$ and $\norm{w_{\epsilon_m} - w_{\epsilon_n}}_{L^2} \leq \norm{w - w_{\epsilon_m}}_{L^2} + \norm{w - w_{\epsilon_n}}_{L^2} \to 0$; Secondly, $\norm{w'_{\epsilon_m} - w'_{\epsilon_n}}_{L^2} = 
\int_{\Theta} \frac{p'(t)^2}{p(t)(1-p(t))} |\mathbf{1}_{\{\epsilon_{m}<p(t)<1-\epsilon_{m} \}} -  \mathbf{1}_{\{\epsilon_n<p(t)<1-\epsilon_n \}}| \to 0$ because $\int_{\Theta} \frac{p'(t)^2}{p(t)(1-p(t))}<\infty$ and dominated convergence theorem. Lastly, each $w_{\epsilon_n}$ is in $W^{1,2}(\Theta)$ and because $W^{1,2}(\Theta)$ is complete, the limit $w$ is also in $W^{1,2}(\Theta)$. Hence, it is without loss of generality to claim $w$ is absolutely continuous. Moreover, $\norm{w' - w_{\epsilon_n}}_{L^2} \to 0$ and $\norm{w_{\epsilon_n}}_{L^2}<\infty$ implies $\norm{w'}_{L^2} = \int_{\Theta} (w'(t))^2 dt <\infty$, $w \in \mathcal{D}^{W}$.
\end{proof}

\subsection{Proof of Lemma \ref{lem:uni} (Convexity)} \label{pf:lem-uni}
\begin{proof}
    Since for each $t$, the second term $u(t)p(t)$ is linear in $p$, we only need to prove the convexity of the following function: 
    \[ \phi(x,y) = \begin{cases}
    \frac{y^2}{x (1-x)} \;  & \text{if } x \in (0,1) \\
    0  \; & \text{if } y=0 \text{ and } x =0 \text{ or } 1 \\
    +\infty  \; & \text{otherwise }.
    \end{cases}\] 
    Notice that its effective domain $\text{dom}\phi:=\{(x,y): \phi(x,y) < + \infty\} = (0,1) \times \mathbb{R} \cup \{0,1\} \times \{0\}$ is a non-empty convex set. We only need to prove the convexity of $\phi$ on its effective domain.

    Consider $(x_1,y_1), (x_2,y_2) \in (0,1) \times \mathbb{R}$, $\lambda \in (0,1)$, and notice that $\phi(x,y) = \frac{y^2}{x}+\frac{y^2}{1-x}$, it suffices to prove that $\phi_1(x,y)=\frac{y^2}{x}$ and $\phi_2(x,y)=\frac{y^2}{1-x}$ are both convex.

    We have: \begin{align*}
        &\phi_1(\lambda x_1 + (1-\lambda)x_2, \lambda y_1 +(1-\lambda)y_2)\\
        =& (\lambda x_1 + (1-\lambda)x_2) (\frac{\lambda y_1 +(1-\lambda)y_2}{\lambda x_1 + (1-\lambda)x_2})^2\\
        =&(\lambda x_1 + (1-\lambda)x_2)[\frac{\lambda x_1 \frac{y_1}{x_1} }{\lambda x_1 + (1-\lambda)x_2} + \frac{(1-\lambda) x_2 \frac{y_2}{x_2} }{\lambda x_1 + (1-\lambda)x_2}]^2\\
        \leq & (\lambda x_1 + (1-\lambda)x_2)[\frac{\lambda x_1 }{\lambda x_1 + (1-\lambda)x_2}(\frac{y_1}{x_1})^2 + \frac{(1-\lambda) x_2 }{\lambda x_1 + (1-\lambda)x_2}(\frac{y_2}{x_2})^2] \\
        =&\lambda \frac{y_1^2}{x_1} +(1-\lambda) \frac{y_2^2}{x_2}
    \end{align*}
    where the inequality is due to Jensen's inequality. Similarily:
    \begin{align*}
        &\phi_2(\lambda x_1 + (1-\lambda)x_2, \lambda y_1 +(1-\lambda)y_2)\\
        =& (\lambda (1-x_1) + (1-\lambda)(1-x_2)) (\frac{\lambda y_1 +(1-\lambda)y_2}{\lambda (1-x_1) + (1-\lambda)(1-x_2)})^2\\
        =&(\lambda (1-x_1) + (1-\lambda)(1-x_2))[\frac{\lambda (1-x_1) \frac{y_1}{1-x_1} }{\lambda (1-x_1) + (1-\lambda)(1-x_2)} + \frac{(1-\lambda) (1-x_2) \frac{y_2}{1-x_2} }{\lambda (1-x_1) + (1-\lambda)(1-x_2)}]^2\\
        \leq & (\lambda (1-x_1) + (1-\lambda)(1-x_2))[\frac{\lambda (1-x_1) }{\lambda (1-x_1) + (1-\lambda)(1-x_2)}(\frac{y_1}{1-x_1})^2 \\ & \quad + \frac{(1-\lambda) (1-x_2) }{\lambda (1-x_1) + (1-\lambda)(1-x_2)}(\frac{y_2}{1-x_2})^2] \\
        =&\lambda \frac{y_1^2}{1-x_1} +(1-\lambda) \frac{y_2^2}{1-x_2}
    \end{align*}

    Now, we consider the case where $(x_2,y_2)=(0,0)$ and $(x_1,y_1) \in (0,1) \times \mathbb{R}$:
    \begin{align*}
    &\phi(\lambda x_1 + (1-\lambda) x_2 , \lambda y_1 + (1-\lambda) y_2 ) = \phi(\lambda x_1 , \lambda y_1 ) \\
    &= \frac{\lambda y_1^2}{x_1 (1- \lambda x_1)} \leq \frac{\lambda y_1^2}{x_1 (1-  x_1)}
    =\lambda \phi(x_1,y_1) +(1-\lambda) \phi(x_2,y_2) ,
    \end{align*}
    and the case where $(x_2,y_2)=(1,0)$ and $(x_1,y_1) \in (0,1) \times \mathbb{R}$:
    \begin{align*}
    &\phi(\lambda x_1 + (1-\lambda) x_2 , \lambda y_1 + (1-\lambda) y_2 ) = \phi(1 - \lambda (1-x_1) , \lambda y_1 ) \\
    &= \frac{\lambda y_1^2}{(x_1 + (1-\lambda)) (1-  x_1)} \leq \frac{\lambda y_1^2}{x_1 (1-  x_1)}
    =\lambda \phi(x_1,y_1) +(1-\lambda) \phi(x_2,y_2),
    \end{align*}
    convexity also hold for the case where $(x_1,y_1), (x_2,y_2) \in \{0,1\} \times \{0\}$ because $\phi(x,y) =0 $ if $y=0$.
\end{proof}

\subsection{Proof of Theorem \ref{thm:nc2} (Necessity)} \label{pf:thm-nc2}
\begin{proof}
 An admissible function $w_{*}$ is said to be a \textit{strong local minimizer} if it is a minimizer within the neighborhood defined as:
\[J^W[w_{*}] \leq J^{W}[w] \; \text{ for } \forall \, w \in  \{ w \in AC([0,1], \mathbb{R}): \norm{w-w_{*}}_{\infty} \leq \epsilon
\} .\] 
Certainly, a global minimizer for the unconstrained problem is also a local strong minimizer and satisfies our necessary condition below, which is from Theorem 18.13 of \cite{C2013}, and we state it here for completeness. 

In the following, $\partial_P \Lambda$ refers to \textbf{proximal subdiferential} of $\Lambda(t,x,v)$ with respect to only $(x,v)$,\footnote{Let $X$ be a normed space, and $f: X \to \mathbb{R} \bigcup \{+\infty\}$, $x \in \text{dom}f$. We say that $\xi \in X^*$ is a \emph{proximal subgradient} of $f$ at $x$ if for some $\sigma \geq 0$ and for some neighborhood $V$ of $x$, we have \[f(y) - f(x) + \sigma(y) \norm{y-x}^2 \geq \langle \xi, y-x \rangle \quad \forall y \in V.\] The \emph{proximal subdifferential} of $f$ at $x$ is the set of proximal subgradient at $x$.} and $\partial_L \Lambda$ refers to \textbf{limiting subdifferential} of $\Lambda(t,x,v)$ with respect to only $(x,v)$.\footnote{The \emph{limiting subdifferentiial} is defined as the set of limits of proximal subgradient in the following sense:
\[\partial_L f(x) =\{ \xi = \lim \xi_i: \xi_i = \partial_P f(x_i), x_i \to x, f(x_i) \to f(x) \}.\]} 
    \begin{theorem}[{\citealt[Theorem~18.13]{C2013}}]
      Consider the minimization problem:
      \[J(x) = l(x(a),x(b)) + \int_{a}^{b} \Lambda(t,x,x'(t)) dt\] over $x: [a,b] \to \mathbb{R}^n$, subject to boundary condition $(x(a),x(b)) \in E$, where $E \subset \mathbb{R}^n \times \mathbb{R}^n$ is closed, the Lagrangian $\Lambda$ is assumed to be finite-valued, \textbf{LB} measurable in $t$ and $(x,v)$\footnote{In the sense that, it is measurable with respect to the $\sigma$-algebra generated by the products of Lebesgue measurable subsets of $[a,b]$ and Borel measurable subsets $\mathbb{R}^2$.} and $l$ is locally Lipschitz. 
      
      If one further suppose that the Lagrangian $\Lambda$ satisfies the following condition: for every bounded subset $S \subseteq \mathbb{R}^n$, there exists a constant $c$ and a $L_1$ integrable function $d$, such that, for almost all $t$, for almost all $(x,v) \in S \times \mathbb{R}^n$, one has:
      \[
      \frac{|\zeta|}{1+|\psi|} \leq c( |v| + |\Lambda(t,x,v)| ) + d(t), \; \;\forall (\zeta,\psi) \in \partial_P \Lambda(t,x,v) 
      \] 
     Let $x_{*}$ be a strong local minimizer of the above problem. There exists an absolutely continuous costate function $q$ satisfying the \textbf{Euler inclusion}:
      \[q'(t) \in \text{co} \{ r: (r,q(t)) \in \partial_L \Lambda(t,x_{*}(t),x_{*}'(t)) \}  \; t \in [a,b]\;\; a.e., \]
      together with the \textbf{Weierstrass condition}: for almost every $t$,
      \[\Lambda(t,x_{*}(t), v) - \Lambda(t,x_{*}(t),x'_{*}(t)) \geq \langle q(t), v-x'_{*}(t) \rangle, \; \forall v \in \mathbb{R}^n\]
      and the \textbf{transversality condition}:
      \[
      (q(a),-q(b)) \in \partial_L l( x_{*}(a), x_{*}(b)) + N_{E}^{L}(x_{*}(a),x_*(b)).
      \]
    \end{theorem}

    First, we verify that the hypothesis of the theorem is satisfied. In our context, there is no endpoint constraint for the value of $w$, so we can set $E$ to be $\mathbb{R} \times \mathbb{R}$ and $l(w(a),w(b)) = 0$ for any $(w(a),w(b))$. Besides, our Lagrangian $\Lambda^W(t,w,w')=(w')^2 - u(t) \frac{1+\sin(w)}{2}$ is finite on $[a,b]$ for each $w \in AC([0,1],[-\frac{\pi}{2},\frac{\pi}{2}])$. It is a \textbf{Carath\'{e}odory function}, as $\Lambda^W$ is Lebesgue measurable in $t$ for any absolutely continuous $w$ and it is continuous in $(w(t),w'(t))$ for each fixed $t$. Hence, we know it is \textbf{LB} measurable by \citealt[Proposition 6.35]{C2013}. \quad Furthermore, for each fixed $t$, we notice that $\Lambda^{W}(t, w, w')$ is smooth in $(w, w')$, so whenever it is non-empty, its \textbf{proximal subdifferential} $\partial_P$ coincides with the its gradient \citep[Corollary 7.32]{C2013}, given by 
    \[( \Lambda^{W}_w(t, w, w'), \; \Lambda^{W}_{w'}(t, w, w'))
    =(\frac{1}{2} u(t) \cos(w),\; 2 w' ),
    \] We can choose $c=0$, $d(t)=|u(t)|$, which is $L_1$ integrable and satisfy the inequality: 
    \begin{align*}
            \frac{|\zeta|}{1+|\psi|} = \frac{\frac{1}{2}|u(t) \cos(w)|}{1 + |2 w'|} \leq |u(t)| = c( |w'| + |\Lambda^W(t,w,w')| ) + d(t)
    \end{align*}
    for $\forall (\zeta,\psi) \in \partial_P \Lambda(t,x,v)$ and all $t \in [a,b]$ and all $(w,w') \in \mathbb{R}^2$. 

    Next, we interpret each of the necessary conditions. 
    
    Since $\Lambda^W(t,w,w')$ is continuously differentiable in $(w,w')$, the \textbf{limiting subdifferential} $\partial_L \Lambda^W$ reduces to its gradient, which is a singleton \citep[Proporsition 11.12]{C2013}. So our \textbf{Euler inclusion} condition actually implies the \textbf{Euler equation}:
    \begin{align*}
        & q'(t) = \Lambda^W_w(t,w,w') =-\frac{1}{2}u(t)\cos(w(t)),  \quad q(t)= \Lambda^W_{w'}(t,w,w') = 2 w'(t) \\
        \Rightarrow & w'(t) = c +\int_0^{t} \frac{-1}{4} u(\tau) \cos(w(\tau)) d\tau \quad  \text{for some constant } c
    \end{align*}
    Besides, the \textbf{Weierstrass condition} is automatically satisfied, as $\Lambda^{W}$ is convex in $w'$:     
    \[\Lambda^W(t,w_{*}(t), v) - \Lambda^W(t,w_{*}(t),w'_{*}(t)) \geq \langle \Lambda^W_{w'}(t,w_{*}(t),w'_{*}(t)), v-w'_{*}(t) \rangle, \; \forall v \in \mathbb{R}\]
    Lastly, since there is no endpoint constraint for $w$, we let $l(w(a),w(b)) \equiv 0$ and $E = \mathbb{R} \times \mathbb{R}$, so the limiting subdifferential is 
    \[ \partial_L l(w_*(a),w_*(b)) = \{(0,0)\}, \] 
    and the limiting normal cone is also a singleton: \[
    N_E^L(w_*(a),w_*(b)) =  \{(0,0)\},
    \]
    Our \textbf{Transversality condition} implies:
    \[
    (q(0), -q(1)) = (2w'(0), -2w'(1) ) = (0,0).
    \]
    Together with the fact that the above costate function $q$ is absolutely continuous, we have:
    \begin{equation}
       \begin{cases}
           w'(t) = -\frac{1}{4} \int_{0}^{t}  u(s) \cos(w(s)) ds & \text{a.e. on } [0,1]\\
           w'(t_{i}^{-}) =  w'(t_{i}^{+}) & \text{ for }= 1,2,\cdots N-1 \\
            w'(0) = w'(1) =0 
       \end{cases} 
    \end{equation}   
    
\end{proof}

\subsection{Proof of Theorem \ref{thm:ext} (Existence)} \label{pf:thm-ext}
\subsubsection{Existence of AC solution}
    We will apply the direct method in calculus of variations to establish the existence of an absolutely continuous solution. The method shares the same spirit of Weierstrass' extreme value theorem, which says that a lower semicontinuous function attains its infimum on a compact support. Compactness comes from the coercive condition and proper choice of the topology, while lower semi-continuity is established by the \textbf{Integral Semicontinuity} lemma from \citet[Theorem 6.38]{C2013}, which we state without proof.

Let $\Omega$ be an open subset of $\mathbb{R}^m$, and $F: \Omega \times \mathbb{R}^l \times \mathbb{R}^n \to \mathbb{R}$ is the integrand, whose arguments are denoted by $(t,x,v)$. We also use $Q \subseteq \Omega \times \mathbb{R}^l$ to describe the constraint on function $x: \Omega \to \mathbb{R}^l$, i.e. the state constraint:  \[(t,x(t)) \in Q \subseteq \Omega \times \mathbb{R}^l  \;\; t \in \Omega \; a.e.\]
We want to show the semicontinuity of the integral functional  \[J(x,v) = \int_{\Omega} F(t,x(t),v(t)) dt .\]
\begin{lemma}[Integral Semicontinuity] \label{thm:int-lsc}
    Suppose the following hypotheses holds:
    \begin{enumerate}[label=(\alph*)]
    \item $F$ is lower semicontinous in $(x,v)$ and convex with respect to $v$;
    \item For every measurable $x: \Omega \to \mathbb{R}^l$ having $(t,x(t)) \in Q$ a.e., and for every measurable $v: \Omega \to \mathbb{R}^n$, the function $t \mapsto F(t,x(t),v(t))$ is measurable;
    \item There exists $\alpha \in L^1(\Omega)$, $\beta \in L^{\infty}(\Omega,\mathbb{R}^n)$ such that \[
    F(t,x,v) \geq \alpha(t) + \langle \beta(t) ,  v\rangle\;\; \forall (t,x) \in Q, v \in \mathbb{R}^n
    \]
    \item Q is closed in $\Omega \times \mathbb{R}^l$.
    \end{enumerate}
    Let $\{ x _i\}$ be a sequence of measurable functions on $\Omega$ satisfying $(t,x(t)) \in Q$ a.e. which converges almost everywhere to a limit $x_*$. Let $\{v_i\}$ be a sequence of functions converging weakly in $L^r(\Omega, \mathbb{R}^n)$ to $v_*$, where $r>1$. Then
    \[
    J(x_*,v_*) \leq \underset{i \to \infty}{\lim \inf} J(x_i,v_i)
    \]
    \end{lemma}

\medskip 
\begin{proof}[Proof of the existence of AC solution]
    Firstly, we notice that the infimum of the value of problem \eqref{pb:w} is finite. This is because our Lagrangian is bounded from below: \[\Lambda^W(t,x,v) = v^2  - u \frac{1+\sin(x)}{2} \geq -\norm{u}_{\infty} > - \infty  .\] Meanwhile, the constant $0$ function achieves a finite $J^W[\mathbf{0}]=0$. Therefore, there exists a sequence of absolutely continuous function $w_i$ such that
    \[ \lim_{i \to \infty} J^W[w] = \underset{w \in AC([0,1],[-\frac{\pi}{2},\frac{\pi}{2}])}{\inf}  J^W[w]\]

    Besides, our Lagrangian satisfies the following coercive condition: 
    \[
    \Lambda^W(t,x,v) = (v)^2  - u(t) \frac{1+\sin(x)}{2}  \geq |v|^2 - \norm{u}_{\infty} \; \; \forall (t,x,v) \in [0,1] \times  [-\frac{\pi}{2},\frac{\pi}{2}] \times \mathbb{R}, \]
    For $i$ large enough, we have:
    \[
      \int_{0}^{1} (|w_i'(t)|^2 -\norm{u}_{\infty}) \, dt  \leq \int_{0}^{1} \Lambda^W(t,w_i(t),w_i'(t)) \, dt \leq \underset{w \in AC([0,1],[-\frac{\pi}{2},\frac{\pi}{2}])}{\inf}  J^W[w]+1
    \]
     This implies that the sequence of derivatives $\{w_i'\}$ is a bounded sequence in $L^2[0,1]$.  Every bounded sequence in $L^2[0,1]$ has a subsequence that converges weakly in $L^2[0,1]$. To see this, note that $L^2[a,b]$ is a reflexive space, and Banach–Alaoglu theorem says that the closed unit ball in the dual space of a normed vector space is compact in the weak-* topology\footnote{
    Let $(X,Y, \langle, \rangle)$ be a pairing of topological vector space, with the duality $(x,y) \mapsto \langle x,y\rangle$ from $X \times Y$ to $\mathbb{R}$. The \textbf{weak topology} on $X$ induced by $Y$ (denoted $\sigma(X,Y)$) is the coarsest topology on $X$ that makes the mappings $\langle \cdot,y \rangle: X \to \mathbb{R}$ continuous for all $y \in Y$. Equivalently, it is characterized by the \textbf{weak convergence} $\overset{w}{\to}$: for any net, $x_{\alpha} \to x \ \Longleftrightarrow\ \langle x_{\alpha}, y \rangle \to \langle x , y \rangle, \;\; \forall y \in Y$; When $Y=X^*$ is the continous dual of $X$, the \textbf{weak-* topology} on $Y$ is the coarsest topology on $Y$ making the mappings $\langle x, \cdot \rangle: Y \to \mathbb{R}$ continuous for all $x \in X$. Equivalently, it is characterized by the \textbf{weak-* convergence} $\overset{w^*}{\to}$: for any net, $y_{\alpha} \to y \ \Longleftrightarrow\ \langle x, y_{\alpha} \rangle \to \langle x , y\rangle, \;\; \forall x \in X$.
    }. Hence, every closed unit ball in $L^2[0,1]$ is also compact in the  weak-* topology on $L^2[0,1]$. Then, by Eberlein-\v{S}mulian Theorem, we know that in a normed space, weak compactness coincides with weak sequential compactness, so a subset $K$ being compact in the weak topology implies that every sequence in $K$ has a subsequence that converges weakly to some point in $K$. Lastly, a bounded sequence is contained by some closed ball in $L^{2}[0,1]$ and must admit a convergent subsequence in $L^{2}[0,1]$.    
    Without loss of generality we can replace the sequence with that subsequence and denote the limit as $v_*$. 
    Besides, we have the unilateral state constraint: 
    \[w(t) \in [-\frac{\pi}{2},\frac{\pi}{2}], \; t \in [0,1]. \] 
    Since $[-\frac{\pi}{2},\frac{\pi}{2}]$ is a compact subset of $\mathbb{R}$, it is also without loss of generality to assume $\lim_{i \to \infty} w_i(0)$ exists.

    Together, we define an absolutely continuous $w_*$ as following:
    \[
    w_*(t) = \lim_{i \to \infty} w_i(0) + \int_{0}^{t} v_*(s) ds, \;\; t \in [0,1],
    \] 
    It can be shown that the sequence $\{w_i\}$ converge pointwise to $w_*$,  by choosing the characteristic function as test function. Let $\mathbf{1}_{(a,t)}(s):=\begin{cases}
        1 & s \in (a,t) \\
        0 & \text{otherwise}
    \end{cases}$, by weak-* convergence:
    \begin{align*}
        &\lim_{i \to \infty} \int_{0}^{1} w_i'(s) \mathbf{1}_{(0,t)}(s) ds = \int_{0}^{1} v_*(s) \mathbf{1}_{(a,t)}(s) ds \\
       \Rightarrow &  \lim_{i \to \infty} w_i(t)=\lim_{i \to \infty} w_i(0) + \int_{0}^{t} w_i'(s) ds = w_*(0) + \int_{a}^{t} v_*(s) ds = w_*(t)
    \end{align*}
    Since each $w_i$ satisfies the state constraint, and $S$ is closed, we know the state constraint is satisfied by $w_*$ at each point: 
    \[w_*(t) =\lim_{i \to \infty} w_i(t)  \in S, \;\; t \in [0,1]. \]

    Then, we will verify that each hypothesis of the \textbf{Integral Semicontinuity} lemma is satisfied.  For (a), it is easy to see our Lagrangian $\Lambda^W(t,x,v)$ is continuous in $(x,v)$ and convex in $v$. \quad For (b), $u$ being piecewise continuous implies $\Lambda^W(t,x,v)$ is measurable in $t$, together with $\Lambda^W(t,x,v)$ being continuous in $(x,v)$, we know the mapping $t \mapsto \Lambda^W(t,w_i(t),w_i'(t))$ is measurable by \citet[Prop. 6.34, 6.35]{C2013}. \quad For (c), we can choose $\alpha(t)=-\norm{u}_{\infty}$ and $\beta(t)=0$, they satisfy
    $\alpha \in L^1([0,1])$, $\beta(x) \in L^{\infty}([0,1])$ and $\Lambda^W(t,x,v) \geq -\norm{u}_{\infty} \; \forall (t,u) \in Q, v \in \mathbb{R}$. \quad For (d),  we take $Q=[0,1] \times [-\frac{\pi}{2},\frac{\pi}{2}]$, and it is closed in $\Omega \times \mathbb{R}$. 

    Lastly, $\{w_i\}$ converges pointwise to $w_*$ implies it converges almost everywhere to $w_*$, and $\{w_i'\}$ also converges weakly to $v_*$ in $L^2[0,1]$,  by Lemma \ref{thm:int-lsc}, we know that the infimum is achieved by $w_*$:
    \[J^W[w_*] \leq \lim_{i \to \infty } J^W[w_i] = \underset{w \in AC([0,1],[-\frac{\pi}{2},\frac{\pi}{2}])}{\inf}  J^W[w].\]
    Since $w_*$ is an absolutely continuous function by construction and satisfies the unilateral state constraint, it is indeed a global minimizer for problem \eqref{pb:w}.
\end{proof}

\subsubsection{Existence of Lipschitz solution}

Next, we establish the existence of a solution in the class of Lipschitz continuous functions through regularity induced by the Euler integral equation:

\begin{corollary}
 Problem \eqref{pb:w} has a solution in $Lip([0,1],[-\frac{\pi}{2},\frac{\pi}{2}])$.
\end{corollary}

\begin{proof}
    Since $u$ is a piecewise continuous function on a compact support, we know it is bounded by some constant $M$ almost everywhere except at finite points. Meanwhile, there exists an absolutely continuous solution that satisfies the Euler Integral Equation:
    \begin{align*}
        & w'(t) = \int_{0}^{t} \frac{1}{4} u(s) \cos(w(s)) ds & \text{a.e. on } [0,1] \\
    \Rightarrow &|w'(s) - w'(t)| = |\int_{t}^{s} u(\tau) \cos(w(\tau))\,d\tau| \leq M |t-s|, 
    \end{align*}
    so we know $w'$ is uniformly bounded on [0,1] and hence $w  \in Lip$.
\end{proof}

\subsubsection{Non-existence of partial boundary solution}
The idea of the proof is to first consider the problem without state constraint, where a global minimizer exists, and the extremals satisfy the set of necessary conditions described in Theorem~\ref{thm:nc2}. We then argue that any candidate that stays at the boundary and then leaves can not be an extremal in the unconstrained problem. Later, we show that any candidate that crosses the boundary can not be a global minimum of the unconstrained problem. Hence, it is a contradiction if a partial boundary candidate can be the global minimizer in the constrained problem. 

\noindent \textbf{(Existence of solution to the unconstrained problem):}
More specifically, we first consider the following optimization problem without state constraint:
\begin{equation}
    \underset{w \in AC(\Theta,\mathbb{R})}{\text{minimize}} \int_{\Theta}  (w'(t))^2  - u(t) \frac{1+\sin(w(t))}{2} \, dt. \tag{relaxed-$W$} \label{pb:relaxed-w}
\end{equation}

We first observe the functional is translation-invariant in the sense that $J^W[w]=J^W[w+2\pi]$. As a result, we may lose compactness due to translation. To account for this, we consider two admissible functions $w_1,w_2$ as equivalent if $w_1(t) - w_2(t)=2k \pi \; a.e.$ for some $k \in \mathbb{Z}$, and pick a representative in each equivalent class that satisfies $w(a) \in [-\pi,\pi)$. The existence proof for an $AC(\Theta,\mathbb{R})$ optimal solution proceeds almost exactly as in the case with state constraint, with the only twist that $[-\pi,\pi)$ is closed in the quotient space $\mathbb{R}/2\pi\mathbb{Z}$ and hence $\lim_{i \to \infty} w_i(a)$ exists and we can construct $w_*$ from the minimizing sequence in the exact same way. 

The solution, as a global minimizer, satisfies the necessary conditions in Theorem~\ref{thm:nc2} and is governed by the following boundary value problem (BVP)
    \begin{equation}
       \begin{cases}
           w''(t) = -\frac{1}{4} u(t) \cos(w_*(t))  & \text{ a.e. on } [0,1]\\
            w'(0) = w'(1) =0 
       \end{cases} 
    \end{equation}  
    
 \noindent \textbf{(Implication of the necessary conditions):} The above equations do not admit solutions that stay at the boundary $\pm \frac{\pi}{2}$ and then leave. To see this, we assume that $w$ is a solution to the BVP and $w(t)=-\frac{\pi}{2}$ on some interval $(t_0,t_1) \subsetneq [0,1]$. Then, by the continuity of $w$ and $w'$, we know $w(t_1)=-\frac{\pi}{2}$ and $w'(t_1)=0$. We get the following initial value problem (IVP)
\begin{equation}
 \begin{aligned}
    & \begin{cases}
         y_1(t) = w(t) \\ 
         y_2(t) = w'(t)
     \end{cases}     \quad
     \begin{cases}
         y_1'(t) = y_2(t) \\ 
         y_2'(t) = -\frac{1}{4} u(t) \cos(y_1(t))
     \end{cases} \\
     & (y_1(t_1), y_2(t_1)) = (-\frac{\pi}{2},0) \\
     \Rightarrow & y' = F(t,y), \quad F(t,y):= (y_2(t),-\frac{1}{4} u(t) \cos(y_1(t))) 
 \end{aligned}
\end{equation}
 We note that $F: [0,1] \times \mathbb{R}^2 \to \mathbb{R}^2$ is a \textbf{Carath\'{e}odory function}: it is Lebesgue measurable in $t$ for each $y$ and continuous in $(y_1,y_2)$ for each $t$. Moreover,  $F$ is Lipschitz in $y$ with a Lipschitz constant independent of $t$ in the strip $[0,1] \times [-\pi,\pi] \times\mathbb{R}$:
 \begin{align*}
     \norm{F(t,y) - F(t,\tilde{y})} &= \norm{(\tilde{y}_2 - y_2,-\frac{1}{4} u(t)(\cos(y_1) -\cos(\tilde{y}_1)) )}  \\
     &=\norm{(\tilde{y}_2 - y_2,-\frac{1}{4} \sin(c) u(t)(y_1 -\tilde{y}_1) )} \text{ for some } c \in (y_1,\tilde{y}_1) \\
     & \leq \max_{ t \in [0,1]} |u(t)| \; \norm{y-\tilde{y}}
 \end{align*}
 By a version of \textbf{Picard-Lindel\"{o}f} Theorem \citep[\S 10.XVIII]{W1998}, we know there exists a unique solution to the IVP. Meanwhile, it is easy to see that $w\equiv -\pi/2$ and $w'\equiv0$ on $[t_1,1]$ solve the IVP. The argument for $[0,t_0]$ and for $w=\frac{\pi}{2}$ is almost identical and we omit it here. Hence, if $w(t)=-\frac{\pi}{2}$ on some interval $(t_0,t_1)$, it must be on the boundary for the whole interval $[0,1]$. Therefore, any extremal of the unconstrained problem must be one of the following three types up to addition or subtraction of multiple of $2\pi$:
     \begin{itemize}
        \item Pure boundary extremal: $w(t) \equiv -\frac{\pi}{2}$ on $[0,1]$ or $w(t) \equiv \frac{\pi}{2}$ on $[0,1]$
        \item Strict interior extremal: $-\frac{\pi}{2}<w(t)<\frac{\pi}{2}$ for all $t \in [0,1]$
        \item Boundary-crossing extremal: $w(t) \in \bigl(-\frac{\pi}{2},\frac{\pi}{2}\bigr)$ on a subset of $[0,1]$ of positive measure, 
    and $w(t) \notin \bigl[-\frac{\pi}{2},\frac{\pi}{2}\bigr]$ on another subset of positive measure, while 
    the contact set $\{t \in [0,1] : |w(t)| = \frac{\pi}{2}\}$ has measure zero.
    \end{itemize}

\noindent \textbf{(Nonexistence of boundary-crossing minimizer):} We proceed to argue that, any boundary-crossing extremal, if it ever exists, can not be a global minimizer of the unconstrained problem. 
Since $J^W$ is invariant under addition of $2 k \pi $, the unconstrained problem is naturally posed on $\mathbb{R}/(2 \pi \mathbb{Z})$. For the compactness argument we identify functions that differ by a constant multiple of $2 \pi$, and select a representative with $w(0) \in [-\pi,\pi)$. In the arguments below, we work with this fixed representative, so all interval comparisons are made in $\mathbb{R}$, not on the quotient.
Recall Assumption~\ref{ass:sc} of $u$ being single crossing, we have the following useful observations: 
\begin{observation}
       If for some integer $k$,  $w \in [-\pi+2k\pi,-\frac{\pi}{2}+2k\pi)$ on some non-trivial interval on $[0,t_u)$ or $w \in (\frac{\pi}{2}+2k\pi,\pi+2k\pi]$ on some non-trivial interval on $(t_u, 1]$, then $w$ can not be a global minimizer. 
\end{observation}
\begin{proof}
    Suppose $k=0$ and $w \in [-\pi,-\frac{\pi}{2})$ on some maximal interval $(t_0,t_1) \subseteq [0,t_u)$ with two possible cases: i) $w(t_0)=w(t_1)=-\frac{\pi}{2}$ and $t_0>0$, or ii) $w(t_1)=-\frac{\pi}{2}>w(t_0)$ and $t_0=0$, such an interval exists because $w$ is continuous, we then construct 
    \begin{equation*} \tilde{w}(t) =
        \begin{cases}
            -\frac{\pi}{2}  & \text{ on } (t_0,t_1) \\
            w & \text{ otherwise}
        \end{cases}
    \end{equation*}
    the new function $\tilde{w}$ is still absolutely continuous and strictly dominates $w$:
    \[J^{W}[w] - J^{W}[\tilde{w}] 
    = \int_{t_0}^{t_1} (w')^2 + \overbrace{\frac{u}{4}}^{<0}\overbrace{(\sin(-\frac{\pi}{2}) -\sin(w))}^{<0}dt >0\]
    Therefore, $w$ can not be a global minimizer. The proof for $w \in (\frac{\pi}{2},\pi]$ on some non-trivial interval on $(t_u, 1]$ is almost identical and the extension to arbitrary $k$ follows from $2\pi$-translation invariance of $J^W$
\end{proof}
\begin{observation}
If up to addition or subtraction of some multiple of $2\pi$, $w(t_u) \notin [-\frac{\pi}{2},\frac{\pi}{2}]$, then $w$ can not be a global minimizer.  
\end{observation}
\begin{proof}
    Suppose $w(t_u) \in [-\pi,-\frac{\pi}{2})$, then by the continuity of $w$, there exist an open interval in $[0,t_u)$ such that $w(t) \in [-\pi, -\frac{\pi}{2})$ on that interval, and according to our above observation, it can not be a global minimizer. The case for $w(t_u) \in (\frac{\pi}{2},\pi]$ is almost identical.
\end{proof}

With these two observation, we only need to consider extremals that satisfy $w(t) \in [-\frac{\pi}{2}, \pi]$ on $[0,t_u)$, $w(t) \in [-\pi,\frac{\pi}{2}]$ on $(t_u,1]$ and $w(t_u) \in [-\frac{\pi}{2},\frac{\pi}{2}]$. Now, suppose that $w \in (\frac{\pi}{2},\pi]$ on some interval in $[0,t_u)$, we will first consider the following reflection:
    \begin{equation*} \tilde{w}(t) =
        \begin{cases}
            \pi -w   & \text{ if $w \in (\frac{\pi}{2},\pi]$ } (t_0,t_1) \\
            w & \text{ otherwise}
        \end{cases}
    \end{equation*}
We know $(w')^2=(\tilde{w}')^2$ almost everywhere, and combining with the symmetry of $\sin(\cdot)$ function around $\frac{\pi}{2}$, we also have $\sin(w) =\sin(\tilde{w})$ for all $t$, this implies $J^{W}[w] =J^{W}[\tilde{w}]$. We will show that $\tilde{w}$ is strictly dominated by the following construction:
\[
\hat{w} = \begin{cases}
    \min\{\tilde{w}, \frac{\tilde{w}(t_u)-\tilde{w}(0)}{t_u} t+\tilde{w}(0)  \} & \text{ on } [0,t_u)\\
    w  & \text{ on } [t_u,1]
\end{cases}
\]
Since affine function $w=A+\frac{B-A}{t_1-t_0} t$ is the solution to the following problem:
\begin{align*}
   & \text{minimize} \int_{t_0}^{t_1} (w'(t))^2 dt \\
   & \text{subject to } \; w(t_0) =A, \quad w(t_1) =B,
\end{align*}
we know $\int_{t_2}^{t_3} (\tilde{w}'(t))^2 - (\hat{w}'(t))^2 dt >0$ within each interval where the two disagree, that is $(t_2, t_3) \in (t_0,t_1)$ where we have $\tilde{w}(t_0) =\hat{w}(t_0)$ and $\tilde{w}(t_1) =\hat{w}(t_1)$ and $\tilde{w} \neq \hat{w}$ on $(t_0,t_1)$.
Meanwhile, we have $-\frac{\pi}{2} \leq \hat{w} \leq \tilde{w} \leq \frac{\pi}{2}$, which implies
$\sin(\hat{w}) \leq \sin(\tilde{w})$ on $[0,t_u]$. We therefore have
\begin{align*}
    &J^{W}[w] - J^{W}[\hat{w}]=  J^{W}[\tilde{w}]- J^{W}[\hat{w}] \\
    &= \int_{0}^{t_u} \overbrace{(\tilde{w}'(t))^2 - (\hat{w}'(t))^2}^{>0} 
    -\overbrace{\frac{u(t)}{4}}^{<0}\overbrace{(\sin(\tilde{w})(t) -\sin(\hat{w})(t))}^{\geq 0} dt >0
\end{align*}
The case for $w \in [-\pi,\frac{\pi}{2})$ on some interval in $(t_u,1]$ is similar, we first reflect the part $[-\pi,\frac{\pi}{2})$ around $-\frac{\pi}{2}$ to construct a $\tilde{w}$, then let \[
\hat{w} = \begin{cases}
    \max\{\tilde{w}, \frac{\tilde{w}(t_u)-\tilde{w}(1)}{t_u-1} t+\tilde{w}(1)  \} & \text{ on } (t_u,1]\\
    w  & \text{ on } [0,t_u]
\end{cases}
\]
Therefore, we have shown that any boundary-crossing extremal can not be a global minimizer.

\noindent \textbf{(Nonexistence of partial boundary minimizer):} Now, suppose that for the constrained problem, the global minimizer is some $w^*$ such that $w^*=\frac{\pi}{2}$ on some strict subset of $[0,1]$. Together with our uniqueness result (Corollary~\ref{lem:uni}), this implies that neither pure boundary extremal or strict interior extremal can be a global minimizer for the unconstrained problem. Hence, by the existence of the global minimizer to the unconstrained problem, we know it must be a boundary-crossing extremal, contradicting to what we have just shown above.

\subsection{Proof of Theorem \ref{thm:reg} (Regularity)} \label{pf:thm-reg}
\begin{proof}
    A constant boundary solution $p_*=0$ or $p_*=1$ is smooth, and the conclusion is trivially satisfied. Now, suppose the solution is not a boundary one, let $w_*=\arcsin{(2p_*-1)}$ be the solution to Problem~(\ref{pb:w}), it has to satisfy the integral Euler equation in \eqref{eq:Euler-NC2}.
    Suppose $u$ is piecewise continuous, then it is also bounded, that is $|u(t)| < M$ for some constant $M$ and for all $t$. Therefore \[
    |w_*'(s) - w_*'(t)| = |-\frac{1}{4}\int_{t}^{s} u(\tau) \cos(w_*(\tau))\,d\tau| \leq M |t-s|,
    \]
    so $w_*'$ is Lipschitz continuous and hence it is differentiable almost everywhere, implying that $w_*$ is twice differentiable almost everywhere.
    
    If we further suppose $u$ is continuous, the fundamental theorem of calculus and continuity of $u(t) \cos(w(t))$ then also implies:
    \[
    w_*''(t) = \lim_{h \to 0} \frac{-\frac{1}{4}\int^{t+h}_{t}  u(\tau) \cos(w_*(\tau)) \, d\tau }{h} = -\frac{1}{4}u(t)\cos(w_*(t)).
    \] 
    Higher order regularity of $w_*$ can be similarly derived from the the identity $w_*''(t) = -\frac{1}{4} u(t) \cos(w_*(t))$.
    Lastly, since the function $\sin(\cdot)$ is smooth, the result hold equally for $p_*$.
\end{proof}

\subsection{Proof of Proposition \ref{prop:monotone}} \label{pf:monotone}
\begin{proof}
    Suppose the solution is a pure boundary solution, then it is a constant and is certainly weakly monotone. Otherwise, if it is a strict interior solution, then the corresponding solution $w_*$ to Problem~\eqref{pb:w} must satisfy the integral Euler equation in \eqref{eq:Euler-NC2}. In this case, $w_*(t) \in (-\frac{\pi}{2},\frac{\pi}{2})$ and $w_*'(t)=\int_{0}^{t} -\frac{u(s)}{4} \cos(w_*(s)) ds$ imply that $w_*'$ is strictly increasing when $u(t)<0$ and strictly decreasing when $u(t)>0$. Combining with the boundary condition $w_*'(0)=w_*'(1)=0$, we know $w_*'(t) \geq 0$ for all $t \in [0,1]$.

    If the prior is uniform on $[0,1]$, then $p_*=\frac{\sin(w_*)+1}{2}$ is also increasing. If the prior is not uniform, we will show that monotonicity is preserved under the normalization described in Appendix~\ref{pf:normalize}. Denote the solution to the normalized problem as $\tilde{p}_*$, and define $G(t)=\int^{t} \frac{1}{\pi(\tau)} d\tau$. We can recover the solution through $p_*(t) = \tilde{p}_*(G(t))$, and $p'_*(t)=\tilde{p}'_*(G(t)) \frac{1}{\pi(t)}$ implies that $p_*(t)$ is monotone as long as $\tilde{p}_*(t)$ is.
\end{proof}

\subsection{Proof of Proposition \ref{prop:S-shape1}} \label{pf:s-shape1}
\begin{proof}
\textbf{Existence of interior solution:} Consider a candidate function $w = \beta \cdot (t -  0.5)$ with $\beta$ close enough to $0$ so that the value of the integral is approximately:
\begin{align*}
    & \int_{0}^{1} \beta^2 - u(t)\frac{1+ \sin( \beta \cdot (t-0.5))}{2} dt \\
    & = \int_{0}^{1} \beta^2 - u(t)\frac{1+ \beta \cdot (t-0.5) + O( \beta^3 \cdot(t-0.5)^3)}{2} dt \\
    &= \beta^2 -\frac{ 1}{2}\int^{1}_{0} u(t) dt - \frac{1}{2} \beta \int_{0}^{1} u(t)  (t-0.5) dt - \frac{1}{2}  O(\beta^3) \int_{0}^{1} u(t)  (t-0.5)^3 dt \\
    & = \beta^2 - \frac{1}{2} \beta \int_{0}^{1} u(t)  (t-0.5) dt  +  O(\beta^3) \leq 0
\end{align*}
in the last equality, the term $\int_0^1 u(t) dt$ is dropped because $u$ is odd-symmetry around $t=\frac{1}{2}$, and the remainder is absorbed in $O(\beta^3)$ because $u$ is bounded.
Moreover, the coefficient $\int_{0}^{1} u(t)  (t-0.5) dt $ is positive because $u(t) \leq 0$ on $[0,0.5)$ and $u(t) \geq 0$ on $(0.5,1]$. Hence, for $\beta>0$ close enough to zero, the term $- \frac{1}{2} \beta \int_{0}^{1} u(t)  (t-0.5) dt$ outweighs $\beta^2$ and the integral achieves a lower value than $0$, which is lower than that given by constant solution $w(t) \equiv \frac{\pi}{2}$ or $w(t) \equiv -\frac{\pi}{2}$. Lastly, Theorem \ref{thm:ext} implies that the optimal solution exists, and since it is not a boundary one, it must be interior. 

\textbf{Sigmoidal shape of the solution:} To see this, consider another function $z(t) = -w_{*}(1-t)$. Its derivatives satisfy:
\[z'(t) = w'_{*}(1-t) \text{ and } z''(t) = - w_{*}''(1-t) ,\]
Substituting into the Euler–Lagrange equation gives:
\[
-\frac{1}{4}\cos(z(t)) u(t) = \frac{1}{4} \cos(w_{*}(1-t))u(1-t) = -w_{*}''(1-t) =z''(t)
\]
which means that $z$ is another solution to the Euler–Lagrange equation. But by the uniqueness of the interior maximizer of Problem \eqref{pb:p}, we must have
\[w_{*}(t) = z(t) = -w_{*}(1-t).\]

Because $u(t) <0$ on $[0,1/2]$ and $w_{*}$ is a minimizer to Problem \eqref{pb:w}, we have $w(1/2) =0$, $w(t) <0$ on $[0,1/2)$ and $w(t) >0$ on $(1/2,1]$.
Recall that $p(t) = \frac{1 + \sin(w_(t))}{2}$, 
\begin{align*}
p''(t) &= \frac{1}{2}[ -\sin(w(t)) (w'(t))^2  + \cos(w(t)) w''(t)] \\
   &= \frac{1}{2}[-\sin(w(t)) (w'(t))^2 - \frac{1}{4} \cos^2(w(t)) u(t)]
\end{align*}
so we have $p_{*}''(t) >0$ $[0,1/2)$ and $p_{*}''(t) <0$ on $(1/2,1]$.
\end{proof}

\subsection{Proof of Proposition \ref{prop:S-shape2}} \label{pf:s-shape2}

   To highlight the main idea of the argument, we first present the proof under a stronger assumption that the payoff function $u$ is absolutely continuous. We then extend the result to the general case where we do not require continuity.
    
    Recall the transformation: $p(t) = \frac{1+\sin(w(t))}{2}$, which implies \begin{align*}
        p' &= \frac{1}{2} \cos(w) w' \\
        p'' &= \frac{1}{2} \cos(w) w'' - \frac{1}{2} \sin(w) (w')^2 \\
        p''' &=\frac{1}{2}( \cos(w) w''' - 3\sin(w)w'w'' -\cos(w) (w')^3 )
    \end{align*}
    
    Meanwhile, we make use of the E-L equation:
    \begin{align*}
        w'' &= -\frac{1}{4} \cos(w) u(t) \\
        w''' &= -\frac{1}{4} \cos(w) u'(t) + \frac{1}{4} u(t) \sin(w) w'
    \end{align*}
    And substitute $w'',w'''$ in the expression of $p'''$:
    \begin{align*}
        p'''=\frac{1}{2} \cos(w) ( u(t)\sin(w) w'  - \frac{1}{4} u'(t) \cos(w) -(w')^3 )
    \end{align*}
    
    By our assumption, there exists $t_u$ such that $u(t) <0$ if $t \in[0,t_u)$ and $u(t) >0$ if $t \in (t_u,1]$. 
   Meanwhile, by our previous result, $u$ being strictly single-crossing implies $w$ is strictly increasing on $[0,1]$ and $w'(t)>0$ on $(0,1)$. We divide our discussion into 3 cases:

   \noindent \textbf{Case I: $w$ crosses 0:}
   If there also exists a $t_w$ such that $w(t) <0$ if $t \in[0,t_w)$ and $w(t) >0$ if $t \in (t_w,1]$.
    
    On $[0,\min\{t_u,t_w\})$, we have 
    \[
    p'' = -\frac{1}{8} \overset{> 0}{\cos^2(w)} \overset{< 0}{u(t)} - \frac{1}{2} \overset{<0}{\sin(w)} \overset{> 0}{(w')^2} > 0
    \]

    Similarly on $(\max\{t_u,t_w\},1]$, we have:
    \[ p'' = -\frac{1}{8} \overset{> 0}{\cos^2(w)} \overset{> 0}{u(t)} - \frac{1}{2} \overset{> 0}{\sin(w)} \overset{> 0}{(w')^2} < 0 \]

    On $(\min\{t_u,t_w\}, \max\{t_u,t_w\})$, we know $u(t)\sin(w(t)) <0$ and hence:
    \[
    p'''=\frac{1}{2} \overbrace{\cos(w)}^{> 0} ( \overbrace{u(t)\sin(w)}^{<0} \overbrace{w'}^{> 0}  - \frac{1}{4} \overbrace{u'(t)}^{\geq 0} \overbrace{\cos(w)}^{> 0} -\overbrace{(w')^3}^{> 0} ) < 0
    \]
    Since $p''(\min\{t_u,t_w\}) \geq 0$ and $p''(\max\{t_u,t_w\}) \leq 0$,  
    and from above we know $p''$ decreases on $(\min\{t_u,t_w\}, \max\{t_u,t_w\})$. Therefore, there exists a $\tau$ such that $p''(t) \geq 0$ if $t \in [0,\tau)$ and  $p''(t) \leq 0$ if $t \in (\tau,1]$.

    \noindent \textbf{Case II: $w$ stays above 0:}
    Suppose $w(t) \geq 0$ for all $t$.
    Firstly we notice that $p''$ is continuous almost everywhere because boundedness of $u$ implies $w'$ is Lipschitz-continuous and $u$ being continuous almost everywhere implies $w''$ is also continuous almost everywhere. Meanwhile, since $w'(0)=0 \Rightarrow p''(0)=-\frac{1}{8} \cos^2(w(0))u(0)>0$, and $w'(1)=0 \Rightarrow p''(1_{-})=-\frac{1}{8} \cos^2(w(1))u(1)<0$, we can assume \[
    t_{p} = \sup\{s: p''(t) >0 \text{ on } [0,s)\}
    \]
  On $[t_u,1]$, we have:
    \[p''=-\frac{1}{8}\overset{> 0}{\cos^2(w)} \overset{> 0}{u(t)} - \frac{1}{2} \overset{> 0}{\sin(w)}\overset{> 0}{(w')^2}<0\]
    this implies that $t_p < t_u$
  On $[0,t_u)$, since $u(t)<0$ and $\sin(w)>0$ we again have: 
   \[
    p'''=\frac{1}{2} \overbrace{\cos(w)}^{> 0} ( \overbrace{u(t)\sin(w)}^{<0} \overbrace{w'}^{> 0}  - \frac{1}{4} \overbrace{u'(t)}^{\geq 0} \overbrace{\cos(w)}^{> 0} -\overbrace{(w')^3}^{> 0} ) < 0
    \]
    So $p''$ decreases on $[0, t_u)$. Therefore, as we have assumed, $p''(t) \geq 0$ if $t \in [0,t_p)$, and $p''(t) \leq 0$ if $t \in (t_p,1]$.
    
    \noindent \textbf{Case III: $w$ stays below 0:}
    On $[0,t_u]$, we have:
    \[p''=-\frac{1}{8}\overset{> 0}{\cos^2(w)} \overset{< 0}{u(t)} - \frac{1}{2} \overset{< 0}{\sin(w)}\overset{> 0}{(w')^2}>0,
    \]
     Similarly, we assume \[
    t_{p} = \sup\{s: p''(t) >0 \text{ on } [0,s)\},
    \] and we must have $t_p>t_u$.
    
    On $(t_u,1]$, $u(t)\sin{(w)} <0$ we have:
    \[
    p'''=\frac{1}{2} \overbrace{\cos(w)}^{> 0} ( \overbrace{u(t)\sin(w)}^{<0} \overbrace{w'}^{> 0}  - \frac{1}{4} \overbrace{u'(t)}^{\geq 0} \overbrace{\cos(w)}^{> 0} -\overbrace{(w')^3}^{> 0} ) < 0
    \]
    Hence, we know $p''(t) \geq 0$ if $t \in [0,t_p)$, and $p''(t) \leq 0$ if $t \in (t_p,1]$.

    More generally, we know $u$ is of bounded variation because it is increasing, and hence can be considered as a Radon measure. By Lebesgue's Decomposition theorem, we can write $u = u_a + u_s$, where $u_a$ is the absolutely continuous part and has derivative almost everywhere on $[0,1]$, while $u_s$ is the singular part, capturing the countable discontinuities. Furthermore, in the sense of the distributional derivative
    \footnote{That is, for any smooth function $\phi$ supported in $[0,1]$, i.e. $\phi \in \mathcal{C}^{\infty}_{c}([0,1])$, we have:
    \begin{align*}
         \int_{0}^{1} \phi(t) du(t) = \int_{0}^{1} \phi(t) u_a'(t)dt + \int_{0}^{1} \phi(t) du_s(t)
    \end{align*}}
    , we have: $du = u'_a dt + du_s$,
    and we know $u'_a \geq 0$ for almost all $t$ and $du_s$ is a non-negative measure on $[0,1]$. 

    Meanwhile, our previous regularity result implies that $w'$ is Lipschitz-continuous and hence $w''$ exists almost everywhere, and we still have:
    \[
    p'' = \frac{1}{2} \cos(w) w''  -\frac{1}{2} \sin{w} (w')^2
    \]
    
    Now, for higher-order derivatives, we have:
    \begin{align*}
            d w'' &= \frac{1}{4} u  \sin{(w)} \, w'\, dt - \frac{1}{4} \cos{(w)} \,du \\
        dp'' &= \frac{1}{2} \cos(w) (u \sin(w) w'dt -\frac{1}{4}\cos(w) du(t) -(w')^3 dt)
    \end{align*}
    since $du$ is a non-negative measure, we know $dp''$ is a non-positive measure restricted on the interval $(\min\{t_u,t_w\}, \max\{t_u,t_w\})$ for the case where $w$ crosses $0$; on the interval $[t_u,1]$ when $w$ stays above $0$; on the interval $[0,t_u]$ when $w$ stays below $0$. Hence, $p''$ is weakly decreasing on the interval in each of the above cases.

\subsection{Proof of Proposition \ref{prop:scale-u}} \label{pf:prop-scale-u}
\begin{proof}
    Denote $\rho_1 =\frac{1}{\lambda_1}$ and $\rho_2 =\frac{1}{\lambda_2}$. We assume that $\rho_1 > \rho_2 >0$.
    
   \noindent \textbf{Proving $\bm{p'_{*,1}(\frac{1}{2})>p'_{*,2}(\frac{1}{2})}$:} By symmetry, we can solve our boundary value problem (BVP) \[\begin{cases}
    w''(t) = -\frac{\rho}{4} u(t) \cos(w(t)) \text{ a.e. on } [0,1] \\
    w'(0) = w'(1) = 0
    \end{cases} \]
    through shooting method with the initial value problem IVP on $[\frac{1}{2},1]$:
    \[\begin{cases}
     w''(t) = -\frac{\rho}{4} u(t) \cos(w(t)) \text{ a.e. on } [\frac{1}{2},1] \\
     w(\frac{1}{2}) =0\\
    w'(\frac{1}{2}) = m
    \end{cases} \]
    that is, we want to find an initial speed $m$ at $t=\frac{1}{2}$ such that the solution to the above IVP satisfy $w'(1)=0$. 
    Again, since we only assume $u$ to be piecewise continuous, the function $F(t,y,\rho)$ may not be continuous in $t$, and the solution is in the sense of \textbf{Carath\'{e}odory}. Our argument relies on results about differentiable dependence of the Carath\'{e}odory solution's in parameter, which can be found in \citet[Theorem 2.1]{SM2000} or \citet[Thorem 3.1]{KS2011}.

    Denote $V(m,\rho) := w'(1;m,\rho)$ to be the terminal speed of the solution corresponding to the IVP with parameter $(m,\rho)$. We will first show that $V$ is strictly increasing in $m$ and strictly decreasing in $\rho$.

    (\textbf{$\bm{V(m,\rho)}$ is strictly increasing in $\bm{m}$:}) Denote $\mu(t):=\frac{\partial}{\partial m} w(t;m,\rho)$, by our differential dependence lemma, we have the following variational equation:
        \[
    \begin{cases}
         \mu(\frac{1}{2})&=0 \\
          \mu'(\frac{1}{2}) &=1 \\
          \mu''(t) &=\frac{\rho}{4} u(t) \sin(w(t;m,\rho))) \mu(t)
    \end{cases}
    \]
    Meanwhile, by our assumption on the payoff that $u(t)\geq 0$ on $[\frac{1}{2},1]$ and our previous result that $w(t;m,\rho) \in [0,\frac{\pi}{2})$ on $[\frac{1}{2},1]$ implies $\sin(w(t;m,\rho)) \geq 0$ on $[\frac{1}{2},1]$. We can prove $\mu'(t)>0$ on $[\frac{1}{2},1]$ by proof of contradiction: Let 
    \[
    \tau :=\sup\{t \in[\frac{1}{2},1]: \mu'(s) >0 \text{ for all } s \in[\frac{1}{2},t] \}
    \]
    we know $\tau \geq \frac{1}{2}$ because $\mu'(\frac{1}{2})=1>0$,
    and now suppose $\tau <1$, then by continuity of $\mu'(t)$ in $t$, we know $\mu'(\tau)=0$. However, on $[\frac{1}{2},\tau]$, we have $\mu'(t)>0$ and $\mu(\frac{1}{2})=0$ together imply $\mu(t) >0$ on $[\frac{1}{2},\tau]$. Therefore, the second derivative $\mu''(t)=\frac{\rho}{4} u(t) \sin(w(t;m,\rho))) \mu(t) \geq 0$ on $[\frac{1}{2},\tau]$ as well, implying $\mu'(\tau) = \mu'(\frac{1}{2})+ \int_{\frac{1}{2}}^{\tau} \mu''(s) ds \geq 1 >0$, a contradiction. As a result, we know $\frac{\partial}{\partial m} V(m,\rho) :=\frac{\partial}{\partial m} w'(1;m,\rho) >0$.
    
    \medskip
    
    (\textbf{$\bm{V(m,\rho)}$ is strictly decreasing in $\bm{\rho}$:}) Denote $\nu(t):=\frac{\partial}{\partial \rho} w(t;m,\rho)$, we have:
    \[
    \begin{cases}
         \nu(\frac{1}{2})&=0 \\
          \nu'(\frac{1}{2}) &= 0 \\
          \nu''(t) &= -\frac{1}{4} u(t) \cos(w(t;m,\rho)) + \frac{\rho}{4} u(t) \sin(w(t;m,\rho)) \nu(t)
    \end{cases}
    \]
    We will prove that $\nu'(t) \leq 0 \text{ for all } s \in[\frac{1}{2},1]$ again by contradiction. Let 
    \[
    \tau :=\inf\{t \in[\frac{1}{2},1]: \nu'(t) > 0 \}
    \] 
    and now suppose $\tau <1$, then by continuity of $\nu'(t)$ in $t$ we know $\nu'(\tau)=0$, and by the definition of infimum we know, for any $\epsilon>0$, there exists some $t \in (\tau,\tau+\epsilon)$ such that $\nu'(t)>0$, in contradiction with the second derivative being negative at $\tau$, i.e. 
    \[ \nu''(\tau) =-\frac{1}{4} \overbrace{u(\tau)}^{\geq 0} \overbrace{\cos(w(\tau;m,\rho))}^{\geq 0} + \frac{\rho}{4} \overbrace{u(\tau)}^{\geq 0}  \overbrace{\sin(w(\tau;m,\rho))}^{\geq 0} \overbrace{\nu(\tau)}^{\leq 0} \leq 0.\]
    We know $\nu(\tau) = \int_{\frac{1}{2}}^{\tau} v'(s) ds \leq 0$ because we assumed $\nu'(s) \leq 0$ on $[\frac{1}{2},\tau)$.
    As a result, we showed $\nu'(t) \leq 0 \text{ for all } s \in[\frac{1}{2},1]$ , hence $\frac{\partial}{\partial \rho} V(m,\rho) :=\frac{\partial}{\partial \rho} w'(1;m,\rho) =\nu'(1) \leq 0$.
    More specifically, as long as $u(t)$ is positive on a positive measure set on $[\frac{1}{2},1]$, we have  $\frac{\partial}{\partial \rho} V(m,\rho) <0$.

    Now, for the solution of the BVP, we know $V(w'_{*,1}(\frac{1}{2}),\rho_1) = V(w'_{*,2}(\frac{1}{2}),\rho_2)=0$, together with our above monotonicity result, $\rho_1 > \rho_2 \Rightarrow w'_{*,1}(\frac{1}{2}) > w'_{*,2}(\frac{1}{2})>0 \Rightarrow  p'_{*,1}(\frac{1}{2})=\frac{1}{2}\cos(0) w'_{*,1}(\frac{1}{2})>  \frac{1}{2}\cos(0) w'_{*,2}(\frac{1}{2})=p'_{*,2}(\frac{1}{2})>0$.

 \bigskip   
\noindent \textbf{Proving $\bm{p_{*,1}(t) < p_{*,2}(t) <\frac{1}{2} \text{ on } [0,\frac{1}{2})}$:} We would again apply the differential dependence result and ideas from maximum principle to the following BVP:
    \[\begin{cases}
     w''(t) = -\frac{\rho}{4} u(t) \cos(w(t)) \text{ a.e. on } [\frac{1}{2},1] \\
     w(\frac{1}{2}) =0\\
    w'(1) = 0
    \end{cases} \] 
    Denote $\nu(t):=\frac{\partial}{\partial \rho} w(t;m,\rho)$, similarly we have:
    \[
    \begin{cases}
         \nu(\frac{1}{2})&=0 \\
          \nu'(1) &= 0 \\
          \nu''(t) &= -\frac{1}{4} u(t) \cos(w(t;m,\rho)) + \frac{\rho}{4} u(t) \sin(w(t;m,\rho)) \nu(t)
    \end{cases}
    \]
    We will prove by contradiction to show that $\nu(t) \geq 0$ on $[\frac{1}{2},1]$. Now suppose there exists $s$ such that $\nu(s)<0$, by continuity, we know there exists some $\tau \in [\frac{1}{2},1]$ where $\min_{t \in [\frac{1}{2},1]} \nu(t) <0$ is achieved at $\tau$. There are two possible cases. 
    
    Case I: if $\tau \in (\frac{1}{2},1)$, then it being a minimum implies $\nu'(\tau)=0$ and $\nu''(\tau)>0$, however, both the term $-\frac{1}{4} u(\tau) \cos(w(\tau;m,\rho))$ and $\frac{\rho}{4} u(\tau) \sin(w(\tau;m,\rho)) \nu(\tau)$ are negative, contradicting to the equality that
    \[ \nu''(\tau) = -\frac{1}{4} u(\tau) \cos(w(\tau;m,\rho)) + \frac{\rho}{4} u(\tau) \sin(w(\tau;m,\rho)) \nu(\tau) \]
    
    Case II: if $\tau=1$, since $\nu'(1)=0$, for $\nu(1)$ to be the minimum,  we need $\nu''(1)\geq 0$. However, by our assumption that $u >0$ on $(1/2,1]$ we have $\nu''(1)<0$ from the equality above, a contradiction. 

    Therefore, we know $\nu(t) \geq 0$ on $[0,1]$ and the inequality is strict when $u$ is non-zero on some set of positive measure. Hence the corresponding solutions to our optimization problem satisfies $w_{*,1}(t)>w_{*,2}(t)>0 \Rightarrow p_{*,1}(t)> p_{*,2}(t)>\frac{1}{2}$ on $(\frac{1}{2},1]$. By symmetry we also have $p_{*,1}(t) < p_{*,2}(t) < \frac{1}{2}$ on $[0,\frac{1}{2})$.
\end{proof}

\section{Existence of interior solution}
In this subsection, we describe a heuristic argument for when the problem admits an interior solution.  Recall the BVP induced by our necessary condition
\begin{equation}
  w''(t) \;=\; -\frac{u(t)}{4}\cos\bigl(w(t)\bigr),
  \qquad w'(0)=w'(1)=0.
  \label{eq:BVP}
\end{equation}
The goal is to decide whether the above equation admits a solution satisfying  $-\frac{\pi}{2}<w(t)<\frac{\pi}{2}$ for all $t$. 
And recall our assumption that $u$ is piecewise continuous, increasing, and
crosses zero once at $t_u\in(0,1)$, so $u<0$ on $[0,t_u)$ and
$u>0$ on $(t_u,1]$.

The idea is to treat the problem as a two-sided shooting and phase-matching problem at $t_u$.
More precisely, for each initial value $w_0 \in[-\pi/2,\pi/2]$ we consider
the forward IVP
\[
 w''(t)=-\tfrac{u(t)}{4}\cos w(t),\quad
 w(0)=w_0,\; w'(0)=0,\qquad t\in[0,t_u],
\]
and for each terminal value $w_1\in[-\pi/2,\pi/2]$ we also consider the backward IVP
\[
 w''(t)=-\tfrac{u(t)}{4}\cos w(t),\quad
 w(1)=w_1,\; w'(1)=0,\qquad t\in[t_u,1].
\]
Because the right-hand side is globally Lipschitz in $(w,w')$, the solution at time $t_u$ depends continuously on the initial value. Varying $w_0$ therefore traces out a continuous curve in the phase plane $(w'(t_u),w(t_u))$, and varying $w_1$ traces out another one. Meanwhile, to keep the solution strictly inside the constraint, we would focus on the segment of the curve with $w(t_u)$ within the boundary.  For valid forward shots we obtain a \emph{forward phase curve}
$\Gamma_-$; valid backward shots give a \emph{backward phase curve} $\Gamma_+$.  An interior solution of \eqref{eq:BVP} corresponds exactly to an intersection point of $\Gamma_-$ and $\Gamma_+$.

\begin{example} 
    To build intuition, let us start with a symmetric, piecewise-constant example where
\[
  u(t) = \begin{cases}
    -800 & t\in[0,\tfrac12),\\
    800  & t\in(\tfrac12,1],
  \end{cases}
\]
The phase curve in $(w'(t_u),w(t_u))$-plane is in Fig . Note that only the highlighted part (red and purple part) is valid as they represent the choice of $w_0$ or $w_1$ that ensure make $w$ not hitting the boundary; and not surprisingly, the two curves intersect at $w(t_u)=0$ by symmetry, this example shows that symmetry of $u$ produces an interior solution.

\begin{figure}[H]
    \centering
    \includegraphics[width=\linewidth]{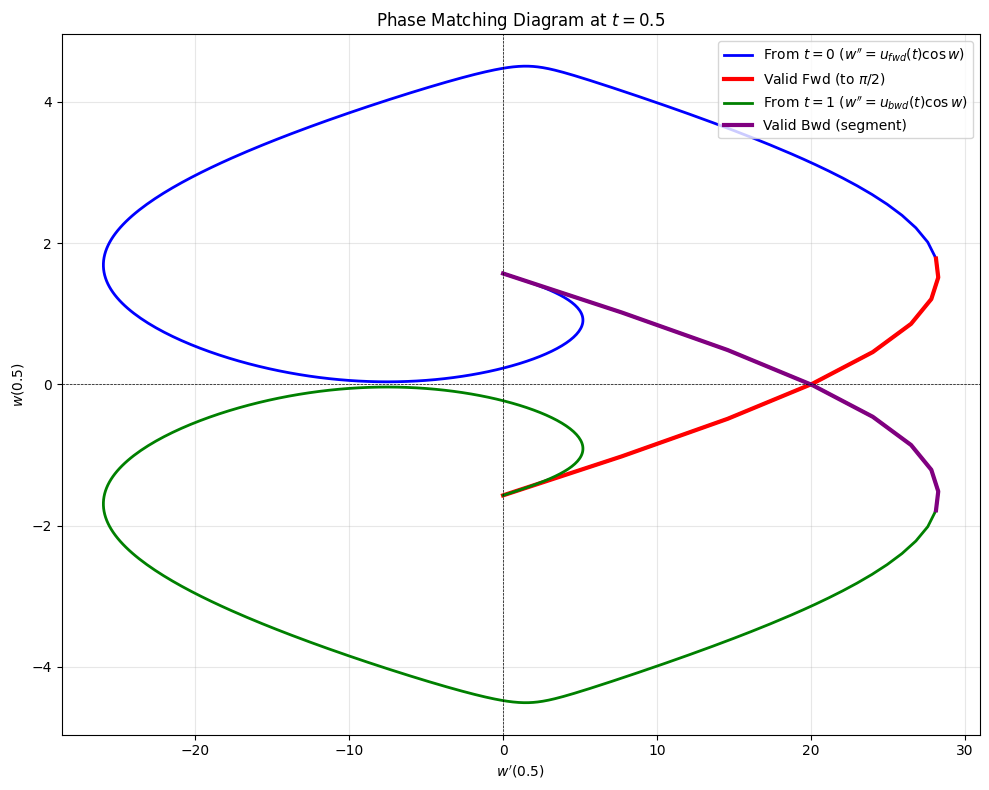}
    \label{fig:ext_eg-1}
\end{figure}

\end{example}

For exploration, we will first focus on the IVP with a constant coefficient:
\[
  w''(t) = A\cos w(t),\qquad
  w(0)=w_1,\; w'(0)=0,\quad t\in[0,\tfrac{1}{2}],
\]
Two qualitatively different regimes appear in the numerics:
\begin{example}
  When $A$ is small (e.g.\ $A=5$), the solution never hits the boundary
  before time $t=\tfrac12$, no matter how close $w_0$ is to $\pi/2$.
  The plot of $w(\tfrac12)$ as a function of $w_0$ approaches $\pi/2$ as
  $w_0\uparrow\pi/2$, but never crosses it, and the entire shooting curve
  stays inside the strip $|w|<\pi/2$.
The upshot is that, if the forcing on a half interval is too weak, then all shots remain interior and the corresponding phase curve does not reach the opposite boundary.

\begin{figure}[H]
    \centering
    \includegraphics[width=\linewidth]{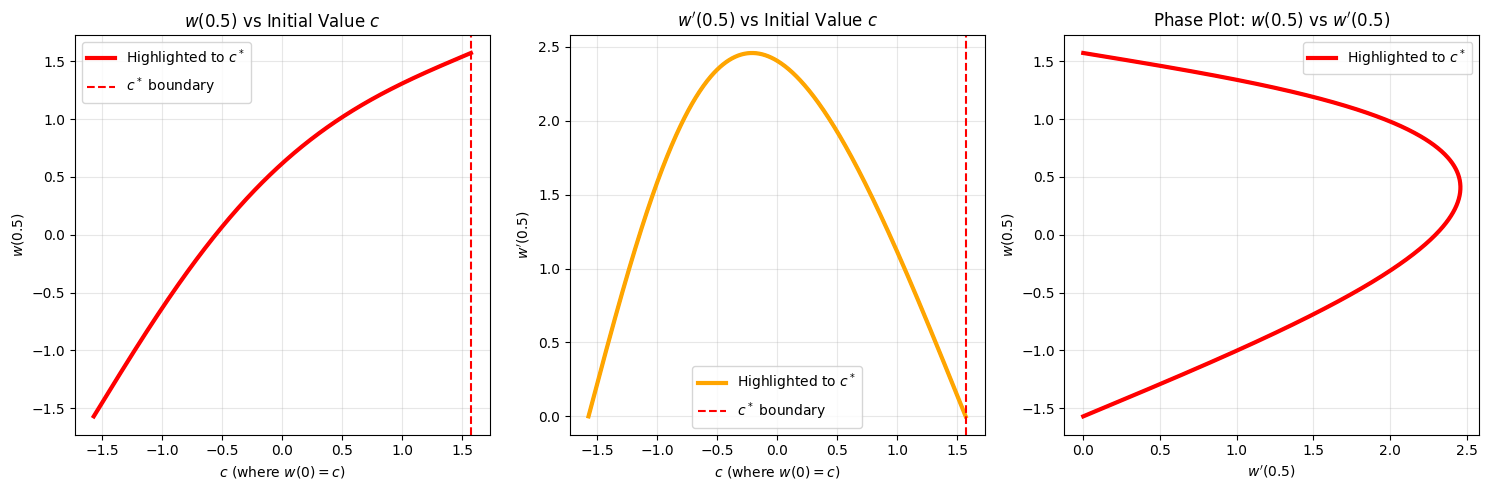}
    \label{fig:ext_eg-2}
\end{figure}
\end{example}

\begin{example}
      When $A$ is large (e.g.\ $A=50$), there exists a critical initial
  value $w_0^\ast<\pi/2$ such that the solution with $w(0)=w_0^\ast$ hits
  the boundary exactly at time $t=\tfrac12$, and any larger $w_0$ leads to
  a boundary hit before $\tfrac12$.
  The valid part of the shooting curve (for $w_0\le w_0^\ast$) now starts
  near $(0,-\pi/2)$ and terminates at a boundary point
  $(v,\pi/2)$ with $v>0$.
  Thus, once the forcing is strong enough, the image of the admissible
  shots along a half interval connect the lower and upper edges of the
  phase strip.

  \begin{figure}[H]
    \centering
    \includegraphics[width=\linewidth]{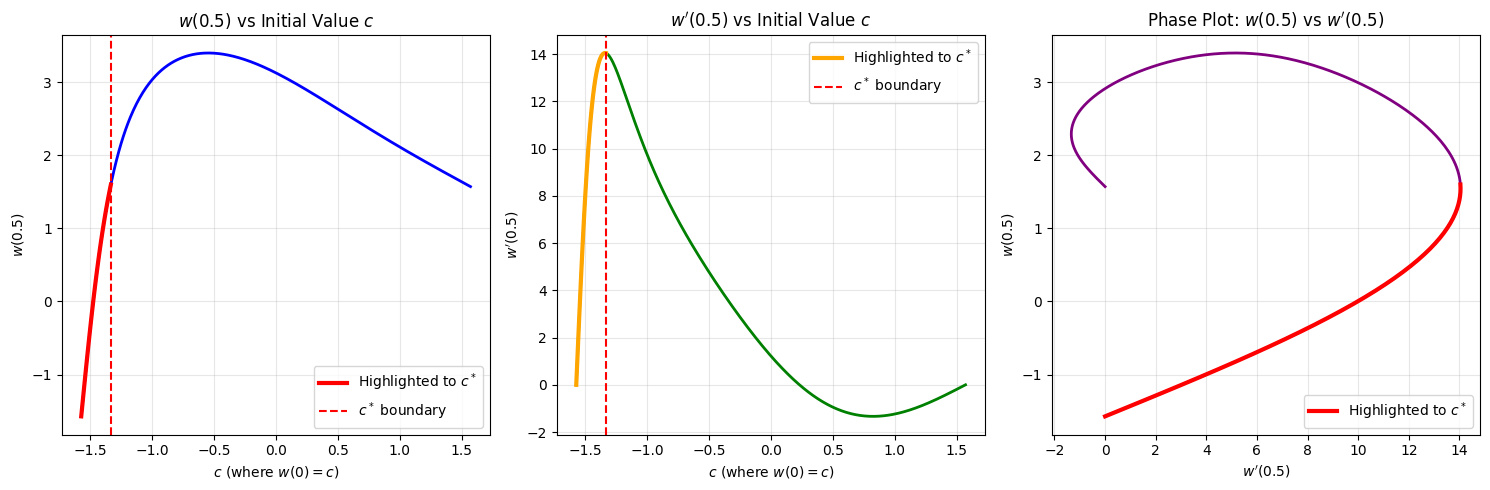}
    \label{fig:ext_eg-2}
\end{figure}
\end{example}

These examples suggest that when both sides are strong,  the valid forward phase curve $\Gamma_-$ connects
  $(0,-\pi/2)$ to some point $(v_f,\pi/2)$ with $v_f>0$, while the valid
  backward phase curve $\Gamma_+$ connects $(0,\pi/2)$ to some point
  $(v_b,-\pi/2)$ with $v_b>0$.
  One curve runs from bottom to top, the other from top to bottom, and both stay in the open strip except at their endpoints.
  By continuity they must cross in the interior, and numerically there
  is a unique intersection point.  This corresponds to an interior
  solution of \eqref{eq:BVP}.
  (See the example with $u=-1200$ on $[0,\tfrac12)$ and $u=80$ on
  $(\tfrac12,1]$.)
  \begin{figure}[H]
    \centering
    \includegraphics[width=\linewidth]{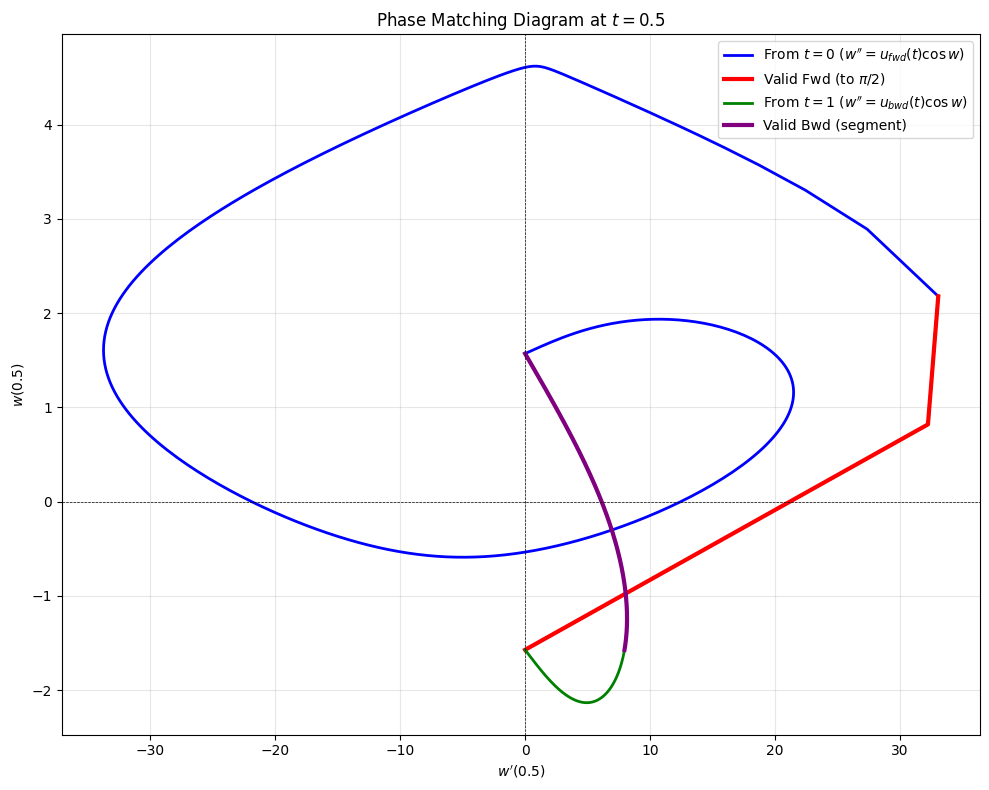}
    \label{fig:ext_eg-2}
\end{figure}

When one side is large while the other is small (e.g.\ $u=-120$ and $u=8$),
  the forward curve $\Gamma_-$ still runs from $(0,-\pi/2)$ towards the
  upper boundary and acquires a substantial horizontal component.
  In contrast, the backward curve $\Gamma_+$ is generated by relatively
  weak forcing on $(\tfrac12,1]$ and stays close to the vertical axis
  $w'(t_u)=0$, connecting $(0,-\pi/2)$ to $(0,\pi/2)$ without moving far
  in the horizontal direction.
  In this configuration the two curves touch at the lower boundary point
  $(0,-\pi/2)$ but need not cross in the interior, and in the numericalexample no interior intersection appears.

  \begin{figure}[H]
    \centering
    \includegraphics[width=\linewidth]{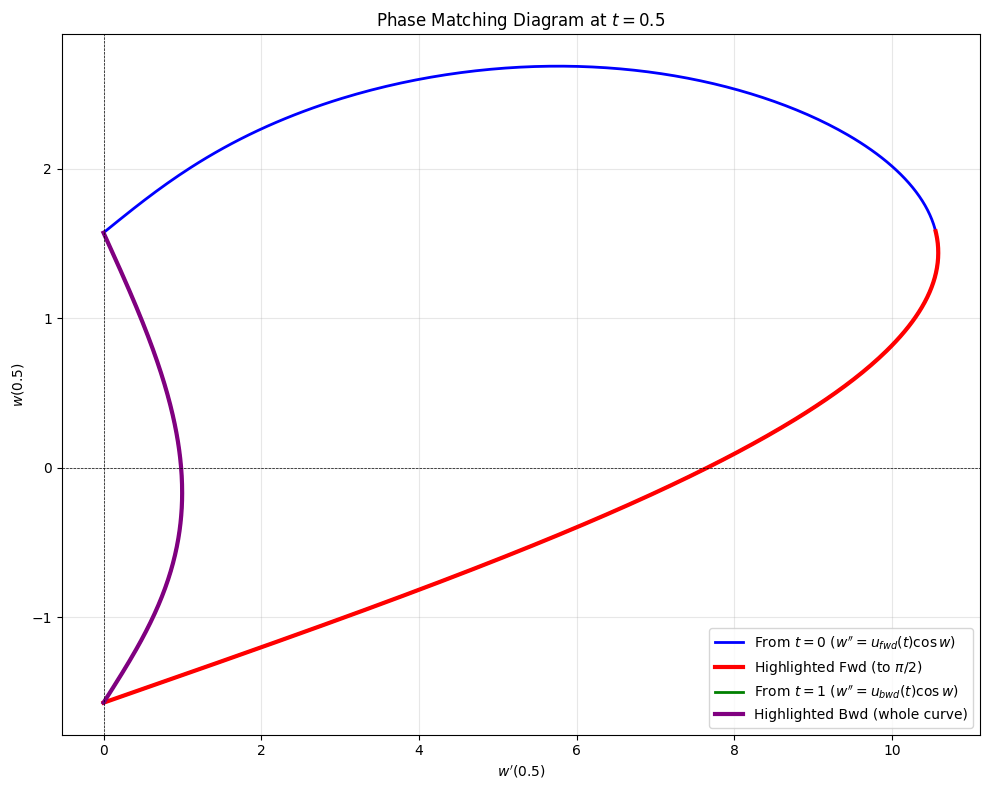}
    \label{fig:ext_eg-2}
\end{figure}

\newpage 
\section{Other proofs}
\subsection{Fisher information cost is Blackwell monotone} \label{pf:Blackwell}
Consider two experiments, both with finite signal space $\sigma: \Theta \to S$ and $\hat{\sigma}: \Theta \to \hat{S}$. Suppose $\sigma$ is more informative than $\hat{\sigma}$ in the sense of \citet{B1951}, that is, there exists a state-independent garbling $K: S \to \Delta(\hat{S})$, such that $\hat{\sigma}(\hat{s}|\theta) = \sum_{s \in S} K(\hat{s}|s) \sigma(s|\theta).$ We will show that experiment $\hat{S}$ is cheaper by showing that the Fisher information is lower point-wise, that is $\mathcal{I}_F( \{ \hat{\sigma}_{\theta} \}) \leq \mathcal{I}_F( \{ \sigma_{\theta} \})$
for each $\theta \in \Theta$.

To see this, recall that the Fisher information is the variance of the \textit{score}, and the score is the partial derivative with respect to $\theta$ of the log-likelihood:
$\frac{\partial}{\partial \theta} \log{\sigma(s|\theta)}$. This is because expectation of the score is $0$:
\[
E_{s|\theta}[\frac{\partial}{\partial \theta} \log{\sigma(s|\theta)}]
=\sum_{s \in S} (\frac{\partial}{\partial \theta} \log{\sigma(s|\theta)}) \sigma(s|\theta) 
=\sum_{s \in S} \frac{\frac{\partial}{\partial \theta} \sigma(s|\theta)}{\sigma(s|\theta)} \sigma(s|\theta) 
=\frac{\partial}{\partial \theta}(\sum_{s \in S} \sigma(s|\theta))=0
\]

Meanwhile, under the garbling, the score for $\hat{\sigma}(\hat{s}|\theta)$ is the conditional expectation of the score for $\sigma(s|\theta)$:
\begin{align*}
    \frac{\partial}{\partial \theta} \log{\hat{\sigma}(\hat{s}|\theta)} 
    &= \frac{\frac{\partial}{\partial \theta}\hat{\sigma}(\hat{s}|\theta) }{\hat{\sigma}(\hat{s}|\theta)} \\
    &= \frac{\partial}{\partial \theta}(\sum_{s \in S} K(\hat{s}|s) \sigma(s|\theta) ) \frac{1}{\hat{\sigma}(\hat{s}|\theta)} \\
    &=\sum_{s \in S}  (\frac{\partial}{\partial \theta}\sigma(s|\theta)) \frac{K(\hat{s}|s) }{\hat{\sigma}(\hat{s}|\theta)}\\
    &=\sum_{s \in S}  \frac{\frac{\partial}{\partial \theta}\sigma(s|\theta)}{\sigma(s|\theta)}
    \frac{K(\hat{s}|s) \, \sigma(s|\theta)}{\hat{\sigma}(\hat{s}|\theta)}\\
    &=\sum_{s \in S}  \frac{\frac{\partial}{\partial \theta}\sigma(s|\theta)}{\sigma(s|\theta)} Pr[s|\hat{s},\theta]\\
    &=\sum_{s \in S} \frac{\partial}{\partial \theta} \log{\sigma(s|\theta)} Pr[s|\hat{s},\theta]=E_{s|\theta}[\frac{\partial}{\partial \theta} \log{\sigma(s|\theta)}|\hat{s}] 
\end{align*}

By the \textbf{Law of total variance}, the Fisher information for $\hat{\sigma}(\theta)$ is:
\begin{align*}
\mathcal{I}_{F}(\hat{\sigma}(\theta)) 
&= \text{Var}_{\hat{s}|\theta}[E_{s|\theta}[\frac{\partial}{\partial \theta} \log{\sigma(s|\theta)}|\hat{s}] ] \\
&=\text{Var}_{s|\theta}[\frac{\partial}{\partial \theta} \log{\sigma(s|\theta)}] - 
E_{\hat{s}|\theta}[\text{Var}_{s|\theta}[\frac{\partial}{\partial \theta} \log{\sigma(s|\theta)}|\hat{s}]] \\
&\leq \text{Var}_{s|\theta}[\frac{\partial}{\partial \theta} \log{\sigma(s|\theta)}] = \mathcal{I}_{F}(\sigma(\theta)) 
\end{align*}

Therefore, the Fisher information at every state decreases under garbling and so does the averaged Fisher information.

\subsection{Fisher information cost is uniformly posterior separable} \label{pf:UPS}
As mentioned in the model section, the cost can be defined over any action space $A$ of finite size, and we start with
\[
C_{FI}(p,\pi) = \int_{\Theta} \pi(t) \sum_{a \in A} \frac{[\frac{d}{dt}\sigma(a|t)]^2}{\sigma(a|t)} dt
\]
where $\pi \in \Delta(\Theta)$ is a prior with full support, and $\sigma(a|t)$ denotes the probability of receiving a signal recommending action $a$ at state $t$. We require $\sigma(a|t)>0$ for each $a \in A$ and almost all $t$, and the actions that are assigned with zero probability across all states are dropped from $A$. 

We denote $\sigma(a)= \int_{\Theta} \sigma(a|t) \pi(t) dt$ as the ex-ante probability of choosing $a$, and $\sigma(\cdot|a) \in \Delta(\Theta)$ as the posterior over the states after receiving a recommendation for action $a$. More specifically, we have $\sigma(t|a) = \frac{\sigma(a|t) \pi(t)}{\sigma(a)}$, and $\sum_a \sigma(a|t)=1$. Also by chain rule \begin{align*}
    \frac{d}{dt}\sigma(t|a) & = \frac{1}{
\sigma(a)}[ \frac{d}{dt}\sigma(a|t) \cdot \pi(t)+ \frac{d}{dt}\pi(t) \cdot \sigma(a|t)] \\
\Rightarrow \frac{d}{dt}\sigma(a|t) & =
\frac{1}{\pi(t)} [ \frac{d}{dt} \sigma(t|a) \cdot \sigma(a) - \frac{d}{dt} \pi(t) \cdot \sigma(a|t) ]
\end{align*}
We have:
\begin{align*}
    C_{FI}(p,\pi) &= \int_{\Theta} \pi(t) \sum_{a \in A} \frac{[\frac{d}{dt}\sigma(a|t)]^2}{\sigma(a|t)} dt \\
    & = \int_{\Theta} \pi(t) \sum_{a \in A} \frac{1}{\sigma(a|t)} \frac{1}{\pi^2(t)} [ \frac{d}{dt} \sigma(t|a) \cdot \sigma(a) - \frac{d}{dt} \pi(t) \cdot \sigma(a|t) ]^2 dt \\
    &= \int_{\Theta}  \sum_{a \in A} \frac{1}{\sigma(a|t)\pi(t) }[(\frac{d}{dt}\sigma(t|a))^2 \sigma^2(a) 
    -2(\frac{d}{dt}\sigma(t|a)) \sigma(a) \frac{d}{dt}\pi(t) \sigma(a|t) \\
    & \quad \quad + (\frac{d}{dt}\pi(t))^2 \sigma^2(a|t)
    ] \\
    & =\int_{\Theta} \sum_a \sigma(a) \frac{(\frac{d}{dt}\sigma(t|a))^2}{\sigma(t|a)}
    - 2\frac{\frac{d}{dt} \pi(t) \cdot \frac{d}{dt} \sigma(t|a)}{\pi(t)} \sigma(t|a)
    + \frac{\sigma(a|t)(\frac{d}{dt} \pi(t))^2}{\pi(t)} dt \\
    &=\int_{\Theta} \sum_a \sigma(a) \frac{(\frac{d}{dt}\sigma(t|a))^2}{\sigma(t|a)}
    - 2\frac{ \sigma(a|t)(\frac{d}{dt}\pi(t))^2}{\pi(t)} 
    - 2(\frac{d}{dt}\pi(t)) (\frac{d}{dt}\sigma(a|t)) 
    + \frac{\sigma(a|t)(\frac{d}{dt} \pi(t))^2}{\pi(t)} dt \\
    &=\sum_a \sigma(a) \int_{\Theta} \frac{(\frac{d}{dt}\sigma(t|a))^2}{\sigma(t|a)} - \frac{(\frac{d}{dt}\pi(t)^2)}{\pi(t)} dt
\end{align*}
The second-to-last equality uses the chain rule to expand the middle term \[ \frac{d}{dt}\sigma(t|a)  = \frac{1}{
\sigma(a)}[ \frac{d}{dt}\sigma(a|t) \cdot \pi(t)+ \frac{d}{dt}\pi(t) \cdot \sigma(a|t)] \]
The last equality holds because 
\[
\int_{\Theta} \sum_a (\frac{d}{dt}\pi(t)) (\frac{d}{dt}\sigma(a|t))  dt = \int_{\Theta} (\frac{d}{dt}\pi(t)) (\frac{d}{dt} \sum_a \sigma(a|t))  dt = \int_{\Theta} (\frac{d}{dt}\pi(t)) (\frac{d}{dt} 1)  dt = 0 
\]
and
\[
\int_{\Theta} \sum_a \frac{\sigma(a|t)(\frac{d}{dt} \pi(t))^2}{\pi(t)} dt = \int_{\Theta} \frac{(\frac{d}{dt} \pi(t))^2}{\pi(t)} dt = \sum_a \sigma(a) \int_{\Theta} \frac{(\frac{d}{dt} \pi(t))^2}{\pi(t)} dt 
\]
So, there exists a potential function $F: \Delta(\Theta) \mathbb{R}$ with the expression:
\[
F(q) = \int_{\Theta} \frac{q'(t)^2}{q(t)}dt
\]
Such a potential function doesn't depend on the prior and is convex because its integrand $\frac{(x')^2}{x}$ is convex.
And we may express Fisher information cost as the expected difference between the potential evaluated at the posterior and prior:
\[C_{FI}(p,\pi) = 
E_{\sigma|\pi}[F(\sigma(\cdot|a)) - F(\pi)] = \sum_a \sigma(a) F(\sigma(\cdot|a))-F(\pi)
\]

\subsection{Justification for the normalizations} \label{pf:normalize}
Our binary choice problem can be formally represented by the tuple $\langle \Theta, u, \pi, \lambda \rangle$ of state space, payoff, prior and scaling parameter:
\[
    \underset{p \in AC(\Theta,[0,1])}{\text{maximize}} \int_{\Theta} \pi(t) u(t) p(t) - \lambda  \pi(t)  \frac{(p'(t))^2}{p(t)(1- p(t))} dt.
\]
Through appropriate rescaling and reparameterization of the payoff function $u(t)$, we can normalize the scaling parameter to $\lambda = 1$
normalize the state space to $[0,1]$. It is without loss of generality in the sense that the solution to any problem of the form $\langle \Theta, u, \pi, \lambda \rangle$ can be recovered from the solution to a normalized problem $\langle [0,1], \tilde{u}, \mathbf{1}, 1 \rangle$, for some appropriately transformed utility function $\tilde{u}$ 
and assume the prior $\pi$ is uniform. We provide justification in the following.

\noindent \textbf{Assume the scaling parameter is $1$:}
Scaling the payoff function by a constant factor $\frac{1}{\lambda}$ results in an equivalent optimization problem, that is, 
\begin{align*}
       \underset{p \in AC(\Theta,[0,1])}{\text{maximize}} \int_{\Theta} \pi(t) u(t) p(t) - \lambda  \pi(t)  \frac{(p'(t))^2}{p(t)(1- p(t))} dt  \\
    \Leftrightarrow \underset{p \in AC(\Theta,[0,1])}{\text{maximize}} \int_{\Theta} \pi(t) (\frac{1}{\lambda}u(t)) p(t) -  \pi(t)  \frac{(p'(t))^2}{p(t)(1- p(t))} dt  
\end{align*} 
the two problems yields the same solution.

\noindent \textbf{Assume the state space is $[0,1]$:} Given an original state space $[a,b]$, we can apply the transformation $y = \frac{t - a}{b - a}$. Solving for $p(t)$ in the original problem $\langle [a,b], u(t), \pi(t), \lambda \rangle$ is equivalent to first solving $\tilde{p}(y) := p((b - a)y + a)$ in the reparameterized problem
$\langle [0,1], \tilde{u}(t):=(b-a)u((b-a)y+a), \tilde{\pi}(t):=\pi((b-a)y+a), \frac{\lambda}{(b-a)^2} \rangle$:
\begin{align*}
\underset{\tilde{p} \in AC([0,1],[0,1])}{\text{maximize}} \int_{[0,1]} [\tilde{\pi}(y) \tilde{u}(y) \tilde{p}(y) - \lambda   \tilde{\pi}(y) \frac{ \frac{1}{(b-a)^2} (\tilde{p}'(y))^2}{\tilde{p}(y)(1- \tilde{p}(y))} ] d((b - a)y) \\
\Leftrightarrow \underset{\tilde{p} \in AC([0,1],[0,1])}{\text{maximize}} \int_{[0,1]} \tilde{\pi}(y) \tilde{u}(y) \tilde{p}(y) - \frac{\lambda}{(b-a)^2}   \tilde{\pi}(y) \frac{ (\tilde{p}'(y))^2}{\tilde{p}(y)(1- \tilde{p}(y))} dy,
\end{align*} 
and then recovering $p(t)$ via $p(t) = \tilde{p}\left(\frac{t - a}{b - a}\right)$.

\noindent \textbf{Assume the prior is uniform:} Since $\pi$ has full support, we can define $G(t) = \int_{0}^{t} \frac{1}{\pi(\tau)} \, d\tau$, which is strictly increasing and invertible. The reparameterization $y = G(t)$ transforms the original problem $\langle [a,b], u, \pi, \lambda \rangle$ into the problem \\ $\langle [0,G(b)], (\pi \circ G^{-1})^2 \cdot u \circ G^{-1}, \frac{\mathbf{1}_{[0,G(b)]}}{G(b)}, \lambda \rangle$, denote $\tilde{p}(y) := p(G^{-1}(y))$, we have:
\begin{align*}
 &(\tilde{p}(y))' =\frac{d}{dy}p(G^{-1}(y)) = \frac{d}{dt} p(t)|_{t=G^{-1}(y)} \frac{1}{\frac{d}{dt} G(t)|_{t=G^{-1}(y)}} = p'(t) \pi(t)|_{t=G^{-1}(y)} \\
 \Rightarrow & \underset{\tilde{p} \in AC([0,G(b)],[0,1])}{\text{maximize}} \int_{0}^{G(b)} [\pi(G^{-1}(y)) u(G^{-1}(y)) \tilde{p}(y) - \lambda  \pi(G^{-1}(y))  \frac{(\frac{1}{\pi(G^{-1}(y))})^2(\tilde{p}'(y))^2}{\tilde{p}(y)(1- \tilde{p}(y))} ]  (\pi(G^{-1}(y)) dy) \\
 \Leftrightarrow & \underset{\tilde{p} \in AC([0,G(b)],[0,1])}{\text{maximize}} \int_{0}^{G(b)} (\pi(G^{-1}(y)))^2 u(G^{-1}(y)) \tilde{p}(y) - \lambda   \frac{(\tilde{p}'(y))^2}{\tilde{p}(y)(1- \tilde{p}(y))} dy
\end{align*}
Hence, we can first solve $\tilde{p}(y)$ from the problem $\langle [0,G(b)], (\pi \circ G^{-1})^2 \cdot u \circ G^{-1}, \frac{\mathbf{1}_{[0,G(b)]}}{G(b)}, \lambda \rangle$ and recover the original solution via $p(t) = \tilde{p}(G(t))$.

\newpage 
\section{From Total Information Cost to Fisher information Cost}\label{app:tic2fic}

The set-up of problem consists of three primitives $(\Theta^N,\gamma^N(\cdot,\cdot),p^N)$, where $\Theta=(\underline{\theta},\bar{\theta})$ is a bounded interval in $\mathbb{R}$ and $\Theta^N=\{\theta_1,\theta_2,\cdots,\theta_N\}$ is a finite discretisation of $(\underline{\theta},\bar{\theta})$ such that $\underline{\theta}<\theta_1<\theta_2 < \cdots <\theta_N < \bar{\theta}$;
$\gamma^{N}: \Theta^N \times \Theta^N \to \mathbb{R}$ takes the form of $\gamma^N(\theta_i,\theta_j) = \frac{\mathbf{1}\{|j-i|=1\}}{(\theta_j - \theta_i)^2}$ to reflect the neighborhood structure consisting of only adjacent pairs;
$p^N \in \Delta(\Theta^N)$ is the prior distribution on the discretisation. 

Throughout, we assume the decision maker acquire information through experiment $(S,\hat{\sigma})$, where $S$ is some fixed signal space, and $\hat{\sigma}: \Theta \to \Delta(S)$ is a map from the state space to distribution on the signal space. We use the short-hand notation $\hat{\sigma}_\theta$ to denote the distribution of signal conditional on the state being $\theta$.

On each decision problem$(\Theta^N,\gamma^N(\cdot,\cdot),p^N)$, we assume the cost of producing signals $\sigma^N$, which is the restriction of $\hat{\sigma}$ to $\Theta^N$, 
takes the total information form: 
\[C^N(\sigma^N)=2\sum_{\theta} p^N(\theta) [\sum_{\theta'} \gamma(\theta,\theta') D_{KL}(\sigma^N_\theta,\sigma^N_{\theta'})]\]

Besides that, we also maintain the following assumptions: 
\begin{assumption} (Partition is finer and finer) \label{assa_1}
    \[\underset{N \to \infty}{\lim} \; \underset{1 \leq\leq N+1}{max}(\theta_i -\theta_{i-1}) =0,\] 
    where $\theta_0 = \underline{\theta}$ and $\theta_{N+1} = \bar{\theta}$.
\end{assumption}

\begin{assumption} (Uniform Convergence of priors) \label{assa_2}\\
    The sequence of priors $p^N$ converges uniformly to some density function $p$ on $(\underline{\theta},\bar{\theta})$ in the following sense:
    \[\lim_{N \to \infty} \max_{0 \leq\leq N} \frac{|p^N(\theta_i) - p(\theta_i)\cdot (\theta_{i+1} - \theta_{i})|}{\theta_{i+1} - \theta_i} =0\]
\end{assumption}

To ensure that the KL-divergence is well-defined, we assume:
\begin{assumption} \label{assa_3}
    (Regularity condition for KL-divergence)\\
    For any $\theta,\theta' \in \Theta$, the signals distributions $\sigma_{\theta}, \sigma_{\theta'}$ are absolutely continuous with respect to each other.
\end{assumption}

To ensure the usual properties of Fisher information still hold, we assume: 
\begin{assumption} (Regularity condition for Fisher-information) \\
The signal structure $\sigma: \Theta \to \Delta(S)$ satisfies the following:
    \begin{enumerate}[label=\alph*)]
        \item $\frac{\partial}{\partial \theta}\sigma(s|\theta)$, $\frac{\partial^2}{\partial \theta^2}\sigma(s|\theta)$, and $\frac{\partial^3}{\partial \theta^3}\sigma(s|\theta)$ exists for all $\theta \in \Theta$, for almost all $s \in S$;
        \item for each $\theta \in \Theta$, there exists Lebesgue-integrable function $F(s)$ and $G(s)$ such that $|\frac{\partial}{\partial \theta}\sigma(s|\theta)|<F(s)$ and  $|\frac{\partial^2}{\partial \theta^2}\sigma(s|\theta)|<G(s)$ for all $s \in S$; 
    \end{enumerate}
\label{assa_4}
\end{assumption}

A key step is to recognize the Fisher information $\int_s [\frac{\partial}{\partial \theta} \log (\hat{\sigma}(s|\theta))]^2 d\sigma(s|\theta)$ as the second derivative of KL-divergence, and the remainder term of such quadratic expansion can be bounded uniformly with a strengthened assumption,
which are stated more formally below:
\begin{lemma}
    Suppose  there exists $\delta>0$, such that for any $\theta$ and for any $\tilde{\theta} \in (\theta-\delta,\theta+\delta)$,
    we have $\int_{S}\frac{\partial^3}{\partial \theta^3}\ln(\sigma(s|\tilde{\theta})) \sigma(s|\theta)ds  < M$, for some constant $M$,
    then
    \begin{align*}
    & \lim_{h \to 0} 
    \sup_{\substack{
    \theta_1,\theta_2 \in \Theta \\
    0< \theta_2 - \theta_1 < h
    }
    } \frac{ D_{KL}(\sigto|\sigtt) -  \frac{1}{2}[\int_s [\frac{\partial}{\partial \theta} \log (\sigma(s|\theta))|_{\theta=\theta_1}]^2 \sigma(s|\theta_1) ds](\theta_2-\theta_1)^2}{(\theta_2-\theta_1)^2} =0
    \end{align*} \label{lema_1}
\end{lemma}

\begin{proof}
    First we consider the Taylor expansion of $\ln(\sigma(s|\theta))$ at $\theta=\theta_1$. By the existence of 3rd order derivative, we have
    \begin{align*} \hspace{-1em}
    \ln(\sigma(s|\theta_2)) = \ln(\sigma(s|\theta_1) )
       + \frac{\partial}{\partial \theta}\ln(\sigma(s|\theta))|_{\theta=\theta_1}(\theta_2-\theta_1)
       +\frac{1}{2} \frac{\partial^2}{\partial \theta^2}\ln(\sigma(s|\theta))|_{\theta=\theta_1}(\theta_2-\theta_1)^2 \\
       +\frac{1}{3!} \frac{\partial^3}{\partial \theta^3}\ln(\sigma(s|\theta))|_{\theta=\theta' \in (\theta_1,\theta_2)}(\theta_2-\theta_1)^3
    \end{align*}
    
    Meanwhile, we have: \begin{align*}
        \frac{\partial}{\partial \theta}\ln(\sigma(s|\theta))|_{\theta=\theta_1} &=
        \frac{ \frac{\partial}{\partial \theta}\sigma(s|\theta)|_{\theta=\theta_1}}{\sigma(s|\theta_1)}\\
        \frac{\partial^2}{\partial \theta^2}\ln(\sigma(s|\theta))|_{\theta=\theta_1} &=
                \frac{ \frac{\partial^2}{\partial \theta^2}\sigma(s|\theta)|_{\theta=\theta_1}}{\sigma(s|\theta_1)} -\frac{(\frac{\partial}{\partial \theta}\sigma(s|\theta)|_{\theta=\theta_1})^2}{\sigma(s|\theta_1)^2}
    \end{align*} 
    Under assumption\hyperref[assa_4]{~\ref*{assa_4}} b)  and dominated convergence theorem, we have for $\forall \theta$:
    \begin{align*}
        \int_{s \in S}  \frac{\partial}{\partial \theta}\sigma(s|\theta)ds& =  
        \int_{s \in S}  \lim_{h \to 0}\frac{\sigma(s|\theta+h) -\sigma(s|\theta)}{h}ds\\
        &=\lim_{h \to 0}\int_{s \in S} \frac{\sigma(s|\theta+h) -\sigma(s|\theta)}{h}ds 
        =\frac{\partial}{\partial \theta}\int_{s \in S}\sigma(s|\theta)ds=0\\
        \int_{s \in S}  \frac{\partial^2}{\partial \theta^2}\sigma(s|\theta)ds& =  
        \int_{s \in S}  \lim_{h \to 0}\frac{\frac{\partial}{\partial \theta}\sigma(s|\theta)|_{\theta+h} -\frac{\partial}{\partial \theta}\sigma(s|\theta)|_{\theta}}{h}ds\\
        &=\lim_{h \to 0}\int_{s \in S} \frac{\frac{\partial}{\partial \theta}\sigma(s|\theta)|_{\theta+h} -\frac{\partial}{\partial \theta}\sigma(s|\theta)|_{\theta}}{h}ds 
        =\frac{\partial}{\partial \theta}(0-0)=0
    \end{align*}

    Plugging-in the above result into the definition of $D_{KL}$, we have:
    \begin{align*}
        \hspace{-2em}
        D_{KL}(\sigto|\sigtt)& = \int_{s \in S} \sigma(s|\theta_1)\ln(\frac{\sigma(s|\theta_1)}{\sigma(s|\theta_2)})ds \\
         &=-\int_{s \in S} \sigma(s|\theta_1) [\frac{ \frac{\partial}{\partial \theta}\sigma(s|\theta)|_{\theta=\theta_1}}{\sigma(s|\theta_1)}(\theta_2-\theta_1)
       -\frac{1}{2} (\frac{ \frac{\partial^2}{\partial \theta^2}\sigma(s|\theta)|_{\theta=\theta_1}}{\sigma(s|\theta_1)} -\frac{(\frac{\partial}{\partial \theta}\sigma(s|\theta)|_{\theta=\theta_1})^2}{\sigma(s|\theta_1)^2})(\theta_2-\theta_1)^2 \\
       &\;-\frac{1}{3!} \frac{\partial^3}{\partial \theta^3}\ln(\sigma(s|\theta))|_{\theta=\theta' \in (\theta_1,\theta_2)}(\theta_2-\theta_1)^3]ds\\
       &=-\underbrace{(\int_{s \in S}\frac{\partial}{\partial \theta}\sigma(s|\theta)|_{\theta=\theta_1}ds)}_{=0} (\theta_2-\theta_1)
       -
       \underbrace{(\frac{1}{2}\int_{s \in S}\frac{\partial^2}{\partial \theta^2}\sigma(s|\theta)|_{\theta=\theta_1}  ds)}_{=0} (\theta_2-\theta_1)^2\\
       &\; +
       (\frac{1}{2} \int_{s\in S} \frac{(\frac{\partial}{\partial \theta}\sigma(s|\theta)|_{\theta=\theta_1})^2}{\sigma(s|\theta_1)}ds)(\theta_2-\theta_1)^2\\
       &\;-
       \frac{1}{3!} (\int_{s\in S}\sigma(s|\theta_1) \frac{\partial^3}{\partial \theta^3}\ln(\sigma(s|\theta))|_{\theta=\theta' \in (\theta_1,\theta_2)}ds)(\theta_2-\theta_1)^3
    \end{align*}

    Now applying the assumption in this lemma, there exists a $\delta>0$ and some constant $M <\infty$ such that
    \begin{align*}
        |\theta_2-\theta_1|<\delta 
     &\Rightarrow 
    \int_{s\in S}\sigma(s|\theta_1) \frac{\partial^3}{\partial \theta^3}\ln(\sigma(s|\theta))|_{\theta=\theta' \in (\theta_1,\theta_2)}ds) < M
    \end{align*}

    Finally, notice that $(\frac{\partial}{\partial \theta} \log (\sigma(s|\theta))^2)\sigma(s|\theta) = \frac{(\frac{\partial}{\partial \theta} \sigma(s|\theta))^2}{\sigma(s|\theta)} $.This allows us to claim that
    \begin{align*}
        &\lim_{(\theta_2-\theta_1) \to 0} |\frac{ D_{KL}(\sigto|\sigto) - [\frac{1}{2}\int_s [\frac{\partial}{\partial \theta} \log (\sigma(s|\theta))]^2 \sigma(s|\theta) ds] (\theta_2 -\theta_1)^2}
        {(\theta_2-\theta_1)^2}| \\
        &\leq \lim_{(\theta_2-\theta_1) \to 0} |M (\theta_2 - \theta_1)|=0
    \end{align*}

    Moreover, since the bound $M$ and the existence of $\delta$ is uniform for all $\theta$ and, we can actually strengthen the conclusion to be:
    \begin{align*}
    & \lim_{h \to 0} 
    \sup_{\substack{
    \theta_1,\theta_2 \in \Theta \\
    0< \theta_2 - \theta_1 < h
    }
    } |\frac{ D_{KL}(\sigto|\sigtt) - [ \frac{1}{2}\int_s [\frac{\partial}{\partial \theta} \log (\sigma(s|\theta))|_{\theta=\theta_1}]^2 \sigma(s|\theta_1) ds](\theta_2 -\theta_1)^2}
    {(\theta_2-\theta_1)^2}| =0
    \end{align*}
    
\end{proof}

\begin{remark}
    The case where the signal has a discrete distribution instead of a continuous distribution can be proven in similar way as above. 

    An example where all the regular conditions are satisfied is when $s|\theta \sim N(\theta,1)$, such that $\ln(\sigma(s|\theta)) = \frac{1}{2}\ln(2\pi) - \frac{(s-\theta)^2}{2}$, $\frac{\partial}{\partial \theta}\ln(\sigma(s|\theta))=-(s-\theta)$, $\frac{\partial^2}{\partial \theta^2}\ln(\sigma(s|\theta))=1$ and $\frac{\partial^3}{\partial \theta^3}\ln(\sigma(s|\theta))=0$.
\end{remark}

\begin{remark}
To see why a stronger result with uniform bound is needed for later proof of the convergence of cost function, consider the case where \[D_{KL}(\sigma_{\theta_1}|\sigma_{\theta_2})= \frac{1}{2}[\int_s [\frac{\partial}{\partial \theta} \log (\sigma(s|\theta))|_{\theta=\theta_1}]^2 \sigma(s|\theta_1) ds](\theta_2-\theta_1)^2 + (\theta_2 - \theta_1)^{2.5},\] we still have $\lim_{(\theta_2-\theta_1) \to 0
    } \frac{ D_{KL}(\sigto|\sigtt) -  \frac{1}{2}[\int_s [\frac{\partial}{\partial \theta} \log (\sigma(s|\theta))|_{\theta=\theta_1}]^2 \sigma(s|\theta_1) ds](\theta_2-\theta_1)^2}{(\theta_2-\theta_1)^2} =0$, but later after taking $\gamma^N(\theta_i,\theta_j) = \frac{\mathbb{1}\{j=i+1\}}{(\theta_j - \theta_i)^2}$, the sum of the remainder term doesn't vanish: $\lim_{N \to \infty}\sum_{i=1}^{N} \frac{(\theta_i - \theta_{i-1})^{2.5}}{(\theta_i - \theta_{i-1})^2}=\lim_{N \to \infty}\sum_{i=1}^{N}(\theta_i - \theta_{i-1})^{\frac{1}{2}}\neq 0$.
\end{remark}

Finally, we are ready to show the Total information cost converges to the Fisher information cost with one more technical assumption. It requires that the integral of Fisher information on the open interval $(\underline{\theta},\bar{\theta})$ can be approximated arbitrarily well by the Riemann sum from some closed sub-interval $[\underline{\theta}+h,\bar{\theta}-h]$:

\begin{assumption} \label{assa_5} (Regularity condition for the Fisher information cost)\\
    The function $J(\theta):=  \int_{s \in S}\sigma(s|\theta) [\frac{\partial}{\partial \theta} \log (\sigma(s|\theta))|_{\theta=\theta_i}]^2 ds$ is $L^1$-integrable with respect to $P(t):= \int_{\underline{\theta}}^{t}p(t)dt$ on $\Theta$ and it is Riemann integrable on $[\underline{\theta}+h,\bar{\theta}-h]$ for any $h>0$, i.e.
    
    \[\int_{\underline{\theta}+h}^{\bar{\theta}-h}|J(\theta)|p(\theta)d\theta < \infty, \]   
and for $\forall h>0$, s.t. for $\forall \epsilon$, there $\exists \delta>0$ s.t. 
    \[|\theta_i -\theta_{i-1}| < \delta \;\; \text{for } i=1,2,\cdots,(N+1)\Rightarrow 
    \sum_{i=1}^{N}|J(\theta_i)(\theta_i -\theta_{i-1}) - \int_{\underline{\theta}+h}^{\bar{\theta}-h}J(\theta)d\theta| < \epsilon\] where $\theta_0 = \underline{\theta}+h$ and $\theta_{N+1} = \bar{\theta}-h$.
\end{assumption}

\begin{theorem}
    Suppose we have a series of decision problem $(\Theta^N,\gamma^N(\cdot,\cdot), p^N)$ that satisfies the above description and Assumption\hyperref[assa_1]{~\ref*{assa_1}} and \hyperref[assa_2]{~\ref*{assa_2}}, suppose the experiment $(S,\hat{\sigma})$ satisfies Assumption \hyperref[assa_3]{~\ref*{assa_3}}, \hyperref[assa_4]{~\ref*{assa_4}} and \hyperref[assa_5]{~\ref*{assa_5}} and those in Lemma\hyperref[lema_1]{~\ref*{lema_1}}, and assume that for in $(\Theta^N,\gamma^N(\cdot,\cdot), p^N)$, the cost for a signal $\sigma^N$, which is the restriction of $\sigma$ to $\Theta^N$, takes the form:\begin{align*}
    C^N(\sigma^N)&=2\sum_{\theta} p^N(\theta) [\sum_{\theta'} \gamma(\theta,\theta') D_{KL}(\sigma^N_\theta,\sigma^N_{\theta'})]
    =2\sum_{i=1}^N
        \frac{p^N(\theta_i)}{(\theta_{i+1} - \theta_i)^2}  D_{KL}(\sigma^N_{\theta_i},\sigma^N_{\theta_{i+1}})
    \end{align*}
    while the cost of the unresstricted $\sigma$ takes the form:
    \[C_{FI}(\sigma)=\int_{\underline{\theta}}^{\bar{\theta}} \int_{s \in S}\sigma(s|\theta) [\frac{\partial}{\partial \theta} \log (\sigma(s|\theta))]^2 ds \; p(\theta) d\theta
    \]
    then the sequence of Total information costs converge to the Fisher information cost in the following sense:
    \[\lim_{N \to \infty}|C^N_{TI}(\sigma^N) - C_{FI}(\sigma) | = 0\]
\end{theorem}

\begin{proof}
    To simplify our notation, we use $J(\theta_i) :=  \int_{s \in S}\sigma(s|\theta) [\frac{\partial}{\partial \theta} \log (\sigma(s|\theta))|_{\theta=\theta_i}]^2 ds$ to denote the Fisher information at $\theta_i$.

    By Lemma\hyperref[lema_1]{~\ref*{lema_1}}, for any $\epsilon>0$, there exists some $\delta_1>0$, we have 
    \[\max_{i}(\theta_i - \theta_{i-1}) < \delta_1 \Rightarrow\sup_{\substack{
    \theta_1,\theta_2 \in \Theta \\
    0< \theta_2 - \theta_1 < \delta
    }
    } |{ D_{KL}(\sigto|\sigtt) - \frac{1}{2} J(\theta_1)(\theta_2 -\theta_1)^2 }| < \epsilon {(\theta_2-\theta_1)^2}\]
    \newline 
    By Assumption\hyperref[assa_5]{~\ref*{assa_5}}, for any $\epsilon>0$, there exists some $\delta_2>0$ such that \[|\int_{\underline{\theta}}^{\underline{\theta}+\delta_2}p(\theta) J(\theta)d\theta| + |\int_{\bar{\theta}-\delta_2}^{\bar{\theta}}p(\theta) J(\theta)d\theta| < \epsilon\]
   By Assumption\hyperref[assa_5]{~\ref*{assa_5}} again, for any $\epsilon>0$, there exists some $\delta_3>0$, such that  
     \[\max_{i}(\theta_i - \theta_{i-1}) < \delta_3 \Rightarrow\sum_{i=1}^{N}|J(\theta_i)(\theta_i -\theta_{i-1}) - \int_{\underline{\theta}+\delta_2}^{\bar{\theta}-\delta_2}J(\theta)d\theta| < \epsilon\] where $\theta_0 = \underline{\theta}+\delta_2$ and $\theta_{N+1} = \bar{\theta}-\delta_2$.

    Finally, by By Assumption\hyperref[assa_1]{~\ref*{assa_1}}, there exists some $M$ such that $N> M \implies \max_{i}(\theta_i - \theta_{i-1}) < \min\{\delta_1,\delta_2, \delta_3\}$ where $\theta_0 = \underline{\theta}$ and $\theta_{N+1} = \bar{\theta}$. 
    We have
    \begin{align*}
        &|C_{TI}^N(\sigma^N) - C_{FI}(\sigma)| \\
        &
        =|2\sum_{i=1}^N
        \frac{p^N(\theta_i)}{(\theta_{i+1} - \theta_i)^2}  D_{KL}(\sigma^N_{\theta_i},\sigma^N_{\theta_{i+1}}) - \int_{\underline{\theta}}^{\bar{\theta}} p(\theta)  J(\theta) d\theta|
        \\
        &
        \leq |
        \sum_{i=1}^{N} p^{N}(\theta_i) J(\theta_i) - \int_{\underline{\theta}}^{\bar{\theta}} p(\theta) J(\theta) d\theta
        |+
        2|\sum_{i=1}^{N} p^{N}(\theta_i)\epsilon(\theta_{i+1} -\theta_{i})|   \qquad {\color{blue}  (\text{Lemma \hyperref[lema_1]{~\ref*{lema_1}}})}\\
        & \leq |
         \sum_{i=1}^{N} p^{N}(\theta_i)J(\theta_i)- \int_{\underline{\theta}}^{\bar{\theta}} p(\theta) J(\theta) d\theta 
        | + 2\epsilon|\sum_{i=1}^{N}(\theta_{i+1} -\theta_{i})| \qquad {\color{blue} (0 \leq p^N(\theta_i) \leq 1) }
        \\
        & \leq |
         \sum_{i=1}^{N} p^{N}(\theta_i)  J(\theta_i)- \int_{\underline{\theta}+\delta_2}^{\bar{\theta}-\delta_2} p(\theta) J(\theta) d\theta 
        |+|(\int_{\underline{\theta}}^{\underline{\theta}+\delta_2}+\int_{\bar{\theta}-\delta_2}^{\bar{\theta}})p(\theta) J(\theta)d\theta|+ 2\epsilon (\bar{\theta} -\underline{\theta})
        \\
        & \leq |\sum_i^N|p^N(\theta_i) -p(\theta_i) (\theta_{i+1} - \theta_{i})| \cdot|J(\theta_i)|+
        |\sum_i^N p(\theta_i)J(\theta_i)(\theta_{i+1} -\theta_i)- \int_{\underline{\theta}+\delta_2}^{\bar{\theta}-\delta_2} p(\theta) J(\theta)d\theta| \\
        &\; + \epsilon (1 + 2(\bar{\theta}-\underline{\theta}))  \qquad \qquad {\color{blue}(\text{Assumption \hyperref[assa_5]{~\ref*{assa_5}}})}
        \\
        & \leq   |\sum_i^N \epsilon(\theta_{i+1} -\theta_i) J(\theta_i)| + \epsilon(2+ 2(\bar{\theta}-\underline{\theta}))  \qquad \qquad {\color{blue}(\text{Assumption\hyperref[assa_2]{~\ref*{assa_2}}})} \\
        &\leq \epsilon(\int_{\bar{\theta}}^{\bar{\theta}}p(\theta) J(\theta)d\theta +2+ 2(\bar{\theta}-\underline{\theta})) \qquad \qquad {\color{blue}(\text{Assumption\hyperref[assa_5]{~\ref*{assa_5}}})}
    \end{align*}
    Lastly, since $\int_{\bar{\theta}}^{\bar{\theta}}p(\theta) J(\theta)d\theta<\infty$ by Assumption\hyperref[assa_5]{~\ref*{assa_5}} and the above $\epsilon$ can be taken arbitrarily small, we show that $\lim_{N \to \infty}|C^N_{TI}(\sigma^N) - C_{FI}(\sigma) | = 0$
    
\end{proof}

\end{document}